\documentclass[11pt,english]{article}
\usepackage{graphicx}
\usepackage{xcolor}
\usepackage{amstext}
\usepackage{caption}
\usepackage{etoolbox}
\usepackage{makeidx}
\include{epsf}
\usepackage{braket}
\usepackage{amsmath,amssymb,epsfig}
\usepackage{amscd}
\usepackage{amsthm}
\usepackage{mathrsfs}
\usepackage{dsfont}
\usepackage{array,boldline,makecell,booktabs}
\usepackage{amsfonts}
\usepackage{authblk} 
\usepackage{sectsty}
\usepackage{amsmath,amssymb,epsfig}
\usepackage{mathtools}
\usepackage{amscd}
\usepackage{amsthm}
\usepackage{mathrsfs}
\usepackage{dsfont}
\usepackage[applemac]{inputenc}
\usepackage[english]{babel}
\usepackage{enumitem} 
\usepackage[]{latexsym}
\usepackage{caption}
\usepackage{hyperref}
\hypersetup{
pdftitle={},%
pdfauthor={},%
pdfsubject={},%
pdfkeywords={},%
colorlinks=true,%
linkcolor=blu,%
citecolor=crimson,%
linktocpage=true,%
pageanchor=true
}

\definecolor{crimson}{rgb}{0.7, 0.08, 0.24}
\definecolor{blu}{rgb}{0.0, 0.18, 0.65} 

\usepackage{graphicx}

\DeclareFontFamily{U}{MnSymbolC}{}
\DeclareSymbolFont{MnSyC}{U}{MnSymbolC}{m}{n}
\DeclareFontShape{U}{MnSymbolC}{m}{n}{
    <-6>  MnSymbolC5
   <6-7>  MnSymbolC6
   <7-8>  MnSymbolC7
   <8-9>  MnSymbolC8
   <9-10> MnSymbolC9
  <10-12> MnSymbolC10
  <12->   MnSymbolC12}{}
\DeclareMathSymbol{\intprod}{\mathbin}{MnSyC}{'270}

\newcommand*{\affaddr}[1]{#1}
\newcommand*{\affmark}[1][*]{\textsuperscript{#1}}

\newtheorem{remark}{Remark}

\newtheorem{theorem}{Theorem}

\newtheorem*{proof*}{Proof}

\newcommand{\be}{\begin{equation}}

\newcommand{\ee}{\end{equation}}
\newcommand{\bes}{\begin{eqnarray}}
\newcommand{\ees}{\end{eqnarray}}
\def\bean{\begin{eqnarray*}}
\def\eean{\end{eqnarray*}}

\newcommand{\dd}{\mathrm{d}}
\newcommand{\gi}{\mathfrak{g}}

\newcommand{\compcent}[1]{\vcenter{\hbox{$#1\circ$}}}
\newcommand{\comp}{\mathbin{\mathchoice
  {\compcent\scriptstyle}{\compcent\scriptstyle}
  {\compcent\scriptscriptstyle}{\compcent\scriptscriptstyle}}}

\renewenvironment{thebibliography}[1]
         {\section*{References}\frenchspacing\small
          \begin{list}{[\arabic{enumi}]}
         {\usecounter{enumi}\parsep=2pt\topsep 0pt
         \settowidth{\labelwidth}{[#1]}
         \leftmargin=\labelwidth\advance\leftmargin\labelsep
         \rightmargin=0pt\itemsep=1pt\sloppy}}{\end{list}}

 \numberwithin{equation}{section}

\input xy
\xyoption{all}

\title{\textbf{Covariant Momentum Map Thermodynamics for Parametrized Field Theories}\vspace{0.25cm}}

\author{
Goffredo Chirco \affmark[1,\&]\footnote{\texttt{goffredo.chirco@aei.mpg.de}}, Marco Laudato \affmark[2]\footnote{\texttt{marco.laudato@graduate.univaq.it}}, and Fabio M. Mele\affmark[3]\footnote{\texttt{fabio.mele@physik.uni-regensburg.de}}\\
\vspace{0.35cm}
\affaddr{\affmark[1]{\normalsize{\textit{Max Planck Institute for Gravitational Physics, Albert Einstein Institute,}}}}\\
\affaddr{\normalsize{\textit{{Am M\"uhlenberg 1, 14476, Potsdam, Germany.\vspace{0.15cm}}}}\\ $^{\&}$Romanian Institute of Science and Technology (RIST), Cluj-Napoca, Romania. }\\
\affaddr{\affmark[2]{\normalsize{\textit{Dipartimento di Ingegneria e Scienze dell'Informazione e Matematica, Universit\'a degli}}}}\\
\affaddr{\normalsize{\textit{{Studi dell' Aquila, Via Vetoio (Coppito 1), 67100 Coppito, L'Aquila, Italy.\vspace{0.15cm}}}}}\\
\affaddr{\affmark[3]{\normalsize{\textit{Institute for Theoretical Physics, University of Regensburg,}}}}\\
\affaddr{\normalsize{\textit{Universit\"atsstra\ss e 31, 93040 Regensburg, Germany.}}}
}

\date{}

\begin{document}

\maketitle

\vspace{-0.25cm}

\begin{abstract}

%


A general-covariant statistical framework capable of describing classical fluctuations of the gravitational field is a thorny open problem in theoretical physics, yet ultimately necessary to understand the nature of the gravitational interaction, and a key to quantum gravity. Inspired by Souriau's symplectic generalization of the Maxwell-Boltzmann-Gibbs equilibrium in Lie group thermodynamics, we investigate a \emph{spacetime-covariant} formulation of statistical mechanics for parametrized first order field theories, as a simplified model sharing essential general covariant features with canonical general relativity. 
Starting from a covariant multi-symplectic phase space formulation, we define a general-covariant notion of Gibbs state in terms of the \emph{covariant momentum map} associated with the lifted action of the diffeomorphisms group on the extended phase space. We show how such a covariant notion of equilibrium encodes the whole information about symmetry, gauge and dynamics carried by the theory, associated to a canonical spacetime foliation, where the covariant choice of a reference frame reflects in a Lie algebra-valued notion of \emph{local} temperature. We investigate how \emph{physical} equilibrium, hence time evolution, emerges from such a state and the role of the gauge symmetry in the thermodynamic description.


\end{abstract}

\newpage
\tableofcontents

\section{Introduction}
In Einstein's theory of General Relativity, spacetime geometry is described as a dynamical field. As any field, spacetime geometry is subject to local laws of propagation and coupling and, in principle, at some scale and within some extreme regime, it may fluctuate. In classical statistical field theory, we are used to describe fields fluctuations over some background geometry, however we currently lack a spacetime-covariant statistical mechanics capable of describing fluctuations of the spacetime itself (hence beyond relativistic statistical mechanics). This is in fact a thorny open issue in theoretical physics, as it requires to reconsider the whole statistical-mechanic formalism in a timeless framework. 
At the heart of the problem lies the conceptual clash of Einstein's general covariant scheme, characterised by systems with vanishing canonical Hamiltonian, with statistical mechanics, whose classical formulation grounds on the notion of Hamiltonian time flow and energy \cite{39}.
Generally covariant theories, indeed, have no notion of distinguished physical time with respect to which everything evolves, while being completely characterised by a vanishing Hamiltonian constraint. 
Accordingly, in General Relativity, spacetime as a field does not evolve, but it can be viewed itself as the evolution of (three dimensional) space, \emph{constrained} as to satisfy the Einstein equation \cite{Giulini}.

The task of this work is to move a step toward  a precise statistical description of spacetime as a fluctuating field, consistently with the general-covariant scheme of General Relativity. Previous work in this direction has focused on the relation between statistical equilibrium and the problem of time in a general-covariant setting, heavily relying on the insights provided by the thermal-time hypothesis \cite{43b,43c,43e, ch, rsm, ko, cko, ck, hr, 19, 19a, kotecha}, and  seeking characterizations of equilibrium, at the classical and quantum level, by means of a covariant extension of the fundamental postulates of ordinary statistical physics. Such results mostly considered simple low dimensional parametrized mechanical systems, while a field-theoretic formulation of the problem has been initiated in \cite{43}.

Here, we concentrate on the peculiarities of the field-theoretic description, while trying to keep conceptually separate the problem of a spacetime covariant definition of statistical state and equilibrium, from its interplay with the emergence of a physical notion of time. Motivated by the \textit{group-covariant} symplectic approach of Souriau's Lie Group Thermodynamics (LGT) \cite{Souriau_b,Souriau,Marle}, we seek a generally covariant extension of the notion of equilibrium statistical states in a multi-symplectic formalism. 

In the LGT framework, the standard notion of time-translation Gibbs equilibrium state is generalized to the Hamiltonian Lie group action of the dynamical symmetry group of the system on its physical phase space. The Gibbs equilibrium state is then derived variationally via maximum entropy principle, by constraining the mean value of the equivariant momentum map associated to the group action. The momentum map plays the role of a generalized Hamiltonian function and the resulting Gibbs state is of equilibrium with respect to any one-parameter subgroup of the Lie group of symmetry of the system under consideration. 

In our approach, we seek a Lie-group generalization of the Gibbs state for generally covariant field theories, with a notion of  equilibrium to be characterized in terms of the Hamiltonian action of the spacetime diffeomorphism group on the extended phase space of the theory.

In particular, we work in the covariant multi-symplectic phase space formalism for generic parametrized first order field theories, which share essential general-covariant features with canonical General Relativity. Here, we identify the relevant covariant macroscopic observable of the generalised phase space with the \emph{covariant momentum map} \cite{gimmsyI, gimmsyII, gimmsyIII, 48,49} associated with the lifted action of the diffeomorphisms group on the \textit{extended phase space}. 

Here, building on the notion of covariant (multi-)momentum map and on the geometric characterization of the parametrization procedure introduced in \cite{lopezgotaymarsdenPFT,lopezgotayPFT}, where diffeomorphisms are promoted to dynamical fields (\textit{covariance fields}), we are able to define a spacetime covariant notion of equilibrium state for parametrized field theories via maximum entropy principle, where the condition of constant mean covariant momentum map is now interpreted as a generalization of the soft (average) diffeomorphism covariance constraint. This is the first result of our work. 

The Hamiltonian character of the diffeomorphism group action on phase space naturally forces us to move from a multi-symplectic to a canonical formalism. Indeed, the differential form character of the covariant momentum map associated to the \textsf{Diff}-action on the multi-phase space of the theory requires us to consider functional densities of fields smeared over co-dimension one slices for the entropy maximization to be well-defined.
 
For the (first order) parametrized field theories, the canonical initial value constraints coincide with the vanishing of the \emph{instantaneous} reduction of the covariant momentum map, the so called energy-momentum map \cite{gimmsyI, gimmsyII, gimmsyIII}, associated to the action of the gauge group of the theory on its extended phase space. Such a momentum map appears to encode all the dynamical information carried by the theory \cite{gimmsyI, gimmsyII, gimmsyIII}. 

In the canonical formalism, first class constraints generating the diffeomorphism symmetry are encoded in the normal and tangential components of the energy-momentum map so that the constraint surface is identified with its zero level set \cite{gimmsyII, gimmsyIII, Giulini, 47}. In the ADM formulation of gravity, for instance, the super-hamiltonian and super-momenta are the components of such energy-momentum map, reflecting the gauge symmetry (diffeomorphism covariance) of the full theory in the instantaneous setting \cite{50,ADM}.

The resulting generalized Gibbs state then defines a canonical distribution over the fields generalised phase space, in which a non-zero weight is assigned to field configurations which do not solve the diffeomorphism covariance constraint.

In absence of boundaries in the spacetime base manifold, the pre-symplectic nature of both the constraint surface and the space of solutions of the field equations prevents us from representing the diffeomorphism group via a Poisson algebra of functions on the on-shell field configurations, as the corresponding co-momentum Hamiltonian would be trivial \cite{torrecovphasespace}. This is what forces us to carry the entire analysis on the generalised (off-shell) phase space, differently from the LGT framework.

Moreover, the representation of the algebra of diffeomorphisms on the parametrized phase space of fields is achieved by taking into account the lifted action of spacetime diffeomorphisms on the embeddings naturally induced by that on the covariance fields. Unlike the standard canonical ADM framework, where lapse and shift are treated as pure phase space-independent Lagrange multipliers and the first-class constraints do not close a genuine Lie algebra, in the extended phase space approach, lapse and shift are functionals of the dynamical embedding fields and the diffeomorphisms algebra can be represented in terms of an equivariant momentum map on the parametrized phase space. Such a dependence is in fact the canonical counterpart of the \textsf{Diff}-action on the auxiliary covariance fields, which in turn ensures the instantaneous momentum map to be equivariant with respect to spacetime diffeomorphisms. The latter property is crucial for the resulting Gibbs state to be of equilibrium with respect to one-parameter subgroups of spacetime diffeomorphisms.



We show how spacetime-covariant Gibbs equilibrium is associated to a canonical spacetime foliation, where the covariant choice of a reference frame reflects in a Lie algebra-valued notion of \emph{local} temperature. Hence, we investigate how \emph{physical} equilibrium, hence time evolution, emerges from such a state and we focus on the role of the gauge symmetry in the thermodynamic description. This is the second main contribution of our work.


The paper is organized as follows. In Sec. \ref{LGT} we briefly recall the main idea of Souriau's Lie Group Thermodynamics and its generalized notion of Gibbs equilibrium state. In Sec. \ref{multisymplFT}, we move to field theory framework and present the main geometric setup on which our analysis is based. After recalling the multi-symplectic formulation for classical field theories (Sec. \ref{multisymplFT}) and the main steps of the parametrization procedure (Sec. \ref{PFTcovfields}), we elaborate this framework for first order parametrized field theories in Sec. \ref{PFTmultimomentmap}, where the notion of covariant multi-momentum map is discussed. In Sec. \ref{canonicalphasespaceEMmap} we move to the corresponding canonical formalism and provide the necessary preliminaries for the representation of spacetime diffeomorphisms via a covariant momentum map on the parametrized phase space of fields discussed in Sec. \ref{diffreprmmap}. With this setting at hand, the extension of Souriau's Lie group themodynamics to parametrized field theories -- that we call \textit{covariant gauge group thermodynamics} -- is derived in Sec. \ref{covariantgibbs}. The covariant notion of Gibbs equilibrium state as well as the corresponding generalized thermodynamic functions are constructed in Sec. \ref{covariantgibbs} and \ref{GTF}, respectively. Finally, in Sec. \ref{EquiandDynam}, we consider some explicit check for our formalism. In particular, in Sec. \ref{gaugefixedeqstate} we investigate how time evolution equilibrium emerges from a suitable gauge-fixing of diffeomorphism symmetry, while in Sec. \ref{thermalinterpret} we discuss analogies and differences with thermal time hypothesis and elaborate on the spacetime interpretation of the thermal flow associated to the covariant Gibbs state. We close with a discussion on the results and the related open perspectives in Sec. \ref{conclusions}.

\section{Souriau's Lie Group Thermodynamics}\label{LGT}

In classical mechanics, a covariant description of dynamics is achieved via symplectic reduction,  by working directly in the space of motions (or space of solutions), and its associated fully constrained phase space. The latter is, by construction, a connected symplectic manifold $(\mathcal M,\omega)$, coordinatised by physical \emph{observable} quantities which are invariant under the dynamical symmetry group of the system (e.g. \cite{torrecovphasespace}).  

On such a space, we shall consider the \emph{Hamiltonian action} $\Psi$ of the connected Lie group $G$, corresponding to the dynamical symmetry group of the system.\footnote{The action $\Psi:G\times\mathcal M\to\mathcal M$ of a Lie group $G$ on a symplectic manifold $(\mathcal M,\omega)$ is said to be symplectic if, for any $g\in G$, $\Psi_g:\mathcal M\to\mathcal M$ is a symplectomorphism, i.e. a smooth diffeomorphism of $\mathcal M$ such that $\Psi_g^*\omega=\omega$. The associated action $\psi:\mathfrak g\to\mathfrak X(\mathcal M)$ of the Lie algebra $\mathfrak g$ of $G$ on $\mathcal M$ is the one parameter group action induced by $\Psi$ via the fundamental vector fields associated to each Lie algebra element $\xi\in\mathfrak g$ satisfying
$$
\psi(\xi)(m)=\frac{\dd}{\dd\lambda}\left(\Psi_{\exp(\lambda\xi)}\cdot m\right)\Bigl|_{\lambda=0}
$$
where $\psi(\xi)\in\mathfrak X(\mathcal M)$ is a vector field on $\mathcal M$ and $m\in\mathcal M$. A symplectic G-action $\Psi$ is said to be Hamiltonian if the associated Lie algebra action $\psi$ is Hamiltonian, that is there exists a smooth real function $f$ on $\mathcal M$ such that $\text{i}_{\psi(\xi)}\omega=-\dd f$.} Let $\mathfrak{g}$ be the Lie algebra of $G$, $\mathfrak{g}^*$ be its dual space, and $J:\mathcal M\to \mathfrak{g}^*$ be a \emph{momentum map} of the $G$-action $\Psi$, that is for any $\xi\in\mathfrak g$ the function $J(\xi):\mathcal M\to\mathbb R$ by $J(\xi)(m)=\braket{J(m),\xi}$ is the Hamiltonian function associated to the vector field $\xi_\mathcal M=\psi(\xi)\in\mathfrak{X}(\mathcal M)$ generating the action $\psi$ of $\mathfrak g$ on $\mathcal M$. A \emph{statistical state} on $(\mathcal M,\omega)$ is a probability law $\mu$ on $\mathcal M$ defined by the product of the Liouville density of $\mathcal M$ 
with a classical distribution function \cite{Souriau_b}
\begin{equation}\label{stat}
\mu(\mathcal A)=\int_{\mathcal A} \rho(m) \,{\omega}^n(m)
\end{equation}
for each Borel subset $\mathcal A$ of $\mathcal M$, with $\rho:\mathcal M \to \mathbb{R}\,( [0, + \infty[)$ being a continuous density function, such that $\int_{\mathcal M} \rho(m)\,{\omega}^n(m)=1$. The Shannon entropy of the state $\mu$ is defined by
\begin{equation}\label{state}
S(\rho)=-\int_{\mathcal M} \rho(m)\,\log{(\rho(m))}\,{\omega}^n(m)\,,
\end{equation}
with $m\, \log(m) = 0$ if $m = 0$, and it measures the amount of uncertainty associated to the probability distribution $\rho(m)$. 
For a given constant mean value of $J$, thermodynamic equilibria $\mu_{eq}$ are states which maximise the entropy \cite{Jaynes, Souriau}, and such that $S(\rho_{eq})$ is stationary with respect to all infinitesimal smooth variations of the probability density.\\
The quantity 
\begin{equation}
Z(\xi) = \int_\mathcal M e^{-\langle J(m),\xi \rangle} \, \omega^n,
\end{equation}
defines a generalized partition function, associated to a Gibbs distribution
\begin{equation}
\rho(\xi) = \frac{1}{Z(\xi)}e^{-\langle J(m),\xi \rangle},
\end{equation}
which is invariant under the action of any one-parameter subgroup of $G$ on $\mathcal M$.
The Hamiltonian function  $\langle J, \xi \rangle: \mathcal M \to \mathbb{R}$, also called \emph{comomentum map}, provides a natural vector-valued generalisation of the Hamiltonian function, retaining the operational information of all conserved charges associated to the symplectic action of a dynamical symmetry group on a the system physical state space. 

As an example, in non-relativistic mechanics, the energy function $E$ on the symplectic manifold of motions $(\mathcal M,\omega)$ is nothing but the momentum map of the Hamiltonian action $\Psi^E$ on the manifold of the one-dimensional Lie group of time translations. As for the standard case of time translations, generalised free energy $F$  and internal energy $Q$ can be defined as smooth functions of the variable $\xi \in \mathfrak{g}$, taking value in $\mathbb{R}$ and in $\mathfrak{g}^*$, respectively. Souriau's group-theoretic formalism realises a covariant generalisation of the standard Maxwell-Boltzmann-Gibbs approach to thermodynamics, in which a (possibly multi-dimensional and non-Abelian) Lie group $G$ acts, by a Hamiltonian action $\Psi$, on the symplectic manifold of motions \cite{Souriau_b,Souriau,Marle}.
This leads, in particular, to a remarkable covariant generalization of Gibbs's equilibrium as soon as we consider the Hamiltonian action of a dynamical symmetry group of the system (e.g. Galilei, Poincar\'e) on its manifold of motions.

But what if the covariant symmetry group of the system is gauge? Building on the notion of covariant momentum map in the multi-symplectic field-theoretic formalism, in the following sections we realise a radical extension of Souriau's Lie group thermodynamics associated to the action of the spacetime diffeomorphism group on the extended phase space of a generic parametrised first order field theory. The ultimate result will be the definition of a generally covariant Gibbs equilibrium state, which encodes the intimate relation between gauge symmetry and dynamics, covariant and canonical formulation.

In the next section, we start with a brief summary of the formalism of multi-symplectic Hamiltonian mechanics for first order parametrised field theories. This formalism will be used throughout the paper to describe the multi-phase space for momentum map thermodynamics. We follow the seminal work\footnote{Further seminal work on the topic, with diverse approaches, includes \cite{Helein0,helein1,Sardanashvily,Ibort1,echeverria-enriquez}.} of \cite{gimmsyI, gimmsyII, gimmsyIII}, to which we refer for an extensive introduction. In particular, to get a fully spacetime-\emph{covariant} multisymplectic description we will extend the multi-phase space description along with the work of \cite{lopezgotaymarsdenPFT,lopezgotayPFT}. Starting from such an \emph{extended} multi-phase space, we will introduce the notion of covariant momentum map and eventually derive the notion of covariant Gibbs equilibrium.

\section{Multisymplectic Formulation of 1st-order Field Theories}\label{multisymplFT}

Let $\mathcal X$ be an oriented $(n+1)$-dimensional manifold (e.g. spacetime), and $\mathcal Y\xrightarrow{\pi_{\mathcal{XY}}}\mathcal X$ a finite-dimensional fiber bundle over $\mathcal X$ with fibers $\mathcal Y_x$ over $x\in\mathcal X$ of dimension $N$. The fiber bundle $\mathcal Y$ is the field theoretic analogue of the configuration space in classical mechanics, or \textbf{configuration bundle}. Physical fields correspond to sections $\phi$ of the bundle. 

A set of local coordinates $(x^\mu,y^A)$ on $\mathcal Y$ consists in $n$+1 local coordinates $x^\mu$, $\mu=0,\dots,n$ on $\mathcal X$, and $N$ fiber coordinates $y^A$, $A=1,\dots,N$, which give the components of the field at a given point $x\in\mathcal X$.
The notion of configuration bundle allows for a finite dimensional formulation of Hamiltonian mechanics for field theories, which relies on a \emph{operational} interpretation of the coordinates of $\mathcal Y$ (i.e., field values and spacetime positions) as partial observables of the theory \cite{Rovelli1,Rovelli2,RovelliPO}.\footnote{To determine an event in a $n$+1-dimensional spacetime, we need one clock and $n$ devices giving us the distance from $n$ reference objects. Accordingly, one needs $N$ devices to measure the components of the field at a given point $x\in\mathcal X$, and $n+1$ devices to determine $x$, thus resulting in a ($n$+$N$+1)-dimensional configuration space. A point in $\mathcal Y$ then represents a correlation between these observables, that is, a possible outcome of a simultaneous measurement of the partial observables.}

The Lagrangian density for a first order classical field theory is given by
\be
\mathscr L:J^1(\mathcal Y)\longrightarrow\Lambda^{n+1}(\mathcal X)\;,
\ee 
where $J^1(\mathcal Y)$ is the first jet bundle\footnote{A $k^{th}$-order Lagrangian density on the $k^{th}$ jet bundle $J^k(\mathcal Y)$ of $\mathcal Y$ for higher order theories \cite{saunders}.} of $\mathcal Y$, the field-theoretic analogue of the tangent bundle of classical mechanics\footnote{In this case $\mathcal Y=\mathbb R\times\mathcal Q$ is the extended configuration space regarded as an $\mathbb R$-bundle over $\mathcal Q$, and $J^1(\mathcal Q \times\mathbb R)$ is isomorphic to the bundle $T\mathcal Q \times T\mathbb R$.}, and $\Lambda^{n+1}(\mathcal X)$ is the space of the $(n+1)$-forms on $\mathcal X$. 

Local coordinates $(x^\mu,y^A)$ on $\mathcal Y$ induce coordinates $v^A_{\mu}$ on the fibers of $J^1(\mathcal Y)$, so that the first jet prolongation $j^1\phi$ of a section $\phi$ of the bundle $\mathcal Y\xrightarrow{\pi_{\mathcal{XY}}}\mathcal X$ gives 
\be
j^1\phi:x^\mu\longmapsto(x^\mu, y^A, v_\mu^A)=(x^\mu, y^A(x), y^A_{,\mu}(x))\;,
\ee
where $y^A_{,\mu}=\partial_\mu y^A$ and $\partial_{\mu}=\partial/\partial x^{\mu}$.

The Lagrangian then reads
\be
\mathscr L(j^1\phi)=L\left(x^{\mu},y^A(x), y^A_{,\mu}(x) \right)\dd^{n+1}x\;, 
\ee
where $\dd^{n+1}x=\dd x^0 \wedge\dots \wedge \dd x^n$ is the volume form on $\mathcal X$.

The introduction of \textit{multimomenta} $p_A^\mu$ and \textit{covariant Hamiltonian} $p$, via the Legendre transform
\be
p_A^\mu=\frac{\partial L}{\partial v^A_\mu}\qquad,\qquad p=L-\frac{\partial L}{\partial v^A_\mu}v^A_\mu\;,
\ee
leads to the field-theoretic analogue of the phase space of classical mechanics, which is provided by the so-called \textbf{multiphase space} $\mathcal Z$. The latter is defined as the sub-bundle of 2-horizontal $(n+1)$-forms on $\mathcal Y$ whose elements can be uniquely written in terms of the fiber coordinates $(p, p^A_{\mu})$ as
\be \label{b3}
z = p \,\dd^{n+1}x + p_A^{\mu} \dd y^A \wedge \dd^n x_{\mu}\;,
\ee
with $\dd^n x_{\mu} = \text{i}_{\partial_{\mu}}\dd^{n+1}x$, such that the contraction $\text{i}_V\text{i}_Wz$ with any two vertical vector fields $V=V^A\partial/\partial y^A$, $W=W^A\partial/\partial y^A$ on $\mathcal Y$ vanishes. As proved in \cite{gimmsyI}, the space $\mathcal Z$ is canonically isomorphic to the dual jet bundle $J^1(\mathcal Y)^*$, the latter playing the role of the cotangent bundle.

In complete analogy to standard symplectic mechanics \cite{ham,marsdenratiu}, the \textbf{canonical Poincar\'e-Cartan $(n+1)$-form} $\Theta$ on $\mathcal Z$ is given by
\be
\Theta = p\, \dd^{n+1}x+p_A^{\mu} \dd y^A \wedge \dd^n x_{\mu}\;, 
\ee
and the \textbf{canonical $(n+2)$-form} $\Omega$ on $\mathcal Z$ is then defined by\footnote{Note that for each field component $y^A$ there are multiple momenta $p_A^{\mu}$, one per each spatiotemporal direction ($\mu=0,\dots n$).}
\be
\Omega= -\dd \Theta=\dd y^A \wedge \dd p_A^{\mu} \wedge \dd^n x_{\mu} -\dd p \wedge \dd^{n+1}x\;. 
\ee
Then the pair $(\mathcal Z,\Omega)$ is an example of \textbf{multisymplectic manifold}\footnote{A multisymplectic manifold $(\mathcal M,\Omega)$ is a manifold endowed with a closed non-degenerate $k$-form $\Omega$ ($k=n+2$ in our case), i.e., such that $\dd\Omega=0$ and $\text{i}_V \Omega\neq0$ for any nonzero tangent vector $V$ on $\mathcal M$ (see e.g. \cite{gimmsyI, helein1,helein2}).}. 

From $(\mathcal Z,\Omega)$, the usual definitions of classical mechanics on the extended phase space are recovered when $\mathcal X$ is one-dimensional (i.e., $n=0$) as reported into the following table which summarizes the analogies between classical symplectic mechanics and the multisymplectic formulation of classical field theories: 

\begin{center}
\begin{tabular}{c|c}
\textbf{Classical Mechanics} ($n=0, \mathcal X\equiv\mathbb R$) & \textbf{Field Theory} ($n>0, \text{dim}\,\mathcal X=n+1$)\\
\hline
&\\
extended configuration space & configuration bundle over spacetime\\
$\mathcal Y=\mathbb R\times\mathcal Q$ & $\mathcal Y\xrightarrow{\pi_{\mathcal{XY}}}\mathcal X$\\
&\\
local coordinates on $\mathcal Y$ & local coordinates on $\mathcal Y$\\
$(t,q^A)$ & $(x^\mu, y^A)$\\
&\\
extended phase space & multiphase space\\
$\mathcal P=T^*\mathcal Y=T^*\mathbb R\times T^*\mathcal Q$ & $J^1(\mathcal Y)^*\cong\mathcal Z\subset\Lambda^{n+1}(\mathcal Y)$\\
&\\
local coordinates on $\mathcal P$ & local coordinates on $\mathcal Z$\\
$(t,q^A, E, p_A)$ & $(x^\mu, y^A, p, p^\mu_A)$\\
&\\
Poincar\'e-Cartan 1-form on $\mathcal P$ & Poincar\'e-Cartan ($n$+1)-form on $\mathcal Z$\\
$\Theta=p_A\dd q^A+E\dd t$ &$\Theta=p\dd^{n+1}x+p_A^\mu\dd y^A\wedge\dd^nx_\mu$\\
&\\
symplectic 2-form on $\mathcal P$ & multisymplectic ($n$+2)-form on $\mathcal Z$\\
$\Omega=\dd q^A\wedge\dd p_A-\dd E\wedge\dd t$ & $\Omega=\dd y^A\wedge\dd p_A^\mu\wedge\dd^nx_\mu-\dd p\wedge\dd^{n+1}x$\\
\end{tabular} 
\end{center}

\subsection{Covariant Multi-Phase Space Description}\label{PFTcovfields}

As it is well known from the pioneering work of Dirac \cite{Diracalgebra1}, further developed by Kucha\v{r} and Isham \cite{ishamkuchar1, kucharcanonquantgrav,kucharcanonquantgencovsystems}, field theories with a fixed background metric can be made generally covariant, i.e., with the spacetime diffeomorphism group as symmetry group, via a \textbf{parametrization} procedure. This amounts to introduce the diffeomorphisms as new dynamical fields as to extend the covariance group of the theory, while leaving the solution space unchanged. A precise geometric reformulation of the parametrization procedure within the context of multi-symplectic field theories
was developed by Castrill\'on L\'opez, Gotay and Marsden in \cite{lopezgotaymarsdenPFT,lopezgotayPFT}. 

The main step of the construction consists in the introduction of the \textbf{covariance fields}, namely (oriented) diffeomorphisms of $\mathcal X$, reinterpreted as sections $\eta:\mathcal X\to\tilde{\mathcal X}$ of the bundle $\tilde{\mathcal X}\times\mathcal X\xrightarrow{\tilde\pi}\mathcal X$, where $\tilde{\mathcal X}$ is a copy the (spacetime) base manifold equipped with a given metric $g$.

Considering $\eta$ as new dynamical fields extends the configuration bundle $\mathcal Y$ as the fibered product $\tilde{\mathcal Y}=\mathcal Y\times_{\mathcal X}(\tilde{\mathcal X}\times\mathcal X)$, whose sections are thought of as pairs $(\phi,\eta)$ according to the diagram
$$
\xymatrix{
\qquad\;\;\,\mathcal Y\times_{\mathcal X}(\tilde{\mathcal X}\times\mathcal X)\ar[dd]^-{\pi_{\mathcal{X}\tilde{\mathcal Y}}}\\
\\
\ar@{-->}@/^3.15pc/[uu]^{\phi} \mathcal X \ar@{-->}@/_3.15pc/[uu]_{\eta}
}
\qquad\qquad \xymatrix{\\ \text{i.e.}\\}\qquad\qquad
\xymatrix{
&\ar[dl]_-{\pi_\ell}\mathcal Y\times_{\mathcal X}(\tilde{\mathcal X}\times\mathcal X)\ar[dd]^-{\pi_{\mathcal{X}\tilde{\mathcal Y}}}\ar[dr]^-{\pi_r}&\\
\mathcal{Y}\ar@/_0.35pc/[dr]_-{\pi_{\mathcal{XY}}} & & *+[r]{\tilde{\mathcal X}\times\mathcal X}\ar@/^0.35pc/[dl]^-{\tilde\pi}\\
&\ar@{-->}@/_0.35pc/[ul]_-{\phi}\mathcal X\ar@{-->}@/^0.35pc/[ur]^-{\eta}&
}
$$
Here $u^a$ ($a=0,\dots,n$) denote the coordinates on $\tilde{\mathcal X}$, while $\pi_\ell:\tilde{\mathcal Y}\to\mathcal Y$ and $\pi_r:\tilde{\mathcal Y}\to\tilde{\mathcal X}\times\mathcal X$ are the projections on the first and second factor of $\tilde{\mathcal Y}$ respectively acting as $\pi_\ell:(x^\mu, y^A, u^a)\mapsto(x^\mu, y^A)$ and $\pi_r:(x^\mu, y^A, u^a)\mapsto(x^\mu, u^\alpha)$, and $\pi_{\mathcal{XY}}(y)=\tilde\pi(u)$.

Finally, one replaces the Lagrangian density $\mathscr L$ of the theory by a new Lagrangian density
\be
\tilde{\mathscr L}:J^1(\tilde{\mathcal Y})\longrightarrow\Lambda^{n+1}(\mathcal X)\;,
\ee
such that
\be\label{tildeL}
\tilde{\mathscr L}(j^1\phi,j^1\eta):=\mathscr L(j^1\phi,\eta^*g)\;.
\ee 
By denoting coordinates on $J^1(\tilde{\mathcal Y})$ by $(x^\mu,y^A,v_\mu^A,u^a, u^a_\mu)$ with $u_\mu^a$ the jet coordinates associated to $u^a$, the new Lagrangian reads 
\be\label{tildeL2}
\tilde{\mathscr L}(x^\mu,y^A,v_\mu^A,u^a, u^a_\mu)=\mathscr L(x^\mu,y^A,v_\mu^A\,; G_{\mu\nu})\;,
\ee
where
\be
G_{\mu\nu}\equiv(\eta^*g)_{\mu\nu}=\eta^a_{,\mu}\,\eta^b_{,\nu}\,g_{ab}	\comp\eta=u^a_{\mu}\,u^b_{\nu}\,g_{ab}\comp\eta\;.
\ee

Now, consider a diffeomorphism of the base manifold, $\alpha_{\mathcal X}\in\textsf{Diff}(\mathcal X)$. We denote by $\alpha_{\mathcal Y}\in\textsf{Aut}(\mathcal Y)$ its lift to $\mathcal Y$. This action is extended to an action on $\tilde{\mathcal Y}$ via bundle automorphisms by requiring that $\textsf{Diff}(\mathcal X)$ acts trivially on $\tilde{\mathcal X}$, i.e.
\be\label{trivext}
\alpha_{\tilde{\mathcal X}}:\,\tilde{\mathcal X}\times\mathcal X\longrightarrow\tilde{\mathcal X}\times\mathcal X\qquad\text{by}\qquad(u,x)\longmapsto(u,\alpha_{\mathcal X}(x))\;.
\ee
The induced action on the space $\tilde{\mathscr Y}\equiv\Gamma(\mathcal X,\tilde{\mathcal Y})$ of sections of $\tilde{\mathcal Y}$ is then given by
\be\label{actiononsections1}
\alpha_{\tilde{\mathscr Y}}(\phi,\eta)=(\alpha_{\mathscr Y}(\phi),\alpha_{\tilde{\mathscr X}}(\eta))\;,
\ee
where
\be\label{actiononsections2}
\alpha_{\mathscr Y}(\phi)=\alpha_\mathcal Y\comp\phi\comp\alpha^{-1}_{\mathcal X} \ ,
\ee
for $\phi\in\mathscr Y\equiv\Gamma(\mathcal X,\mathcal Y)$, generalizes the usual push-forward action on tensor fields, and
\be\label{actiononsections3}
\alpha_{\tilde{\mathscr X}}(\eta)=\eta\comp\alpha^{-1}_{\mathcal X}\;,
\ee
is the (left) action by composition on sections of the trivial bundle $\tilde{\mathcal X}\times\mathcal X$. The modified field theory on $J^1(\tilde{\mathcal Y})$ with Lagrangian density given in \eqref{tildeL} is \textbf{$\textsf{Diff}(\mathcal X)$-covariant}. Indeed, the Lagrangian density \eqref{tildeL} is \textbf{$\textsf{Diff}(\mathcal X)$-equivariant}, i.e. \cite{lopezgotaymarsdenPFT,lopezgotayPFT}
\be\label{Gequiv}
\tilde{\mathscr L}\left(j^1(\alpha_{\mathscr Y}(\phi)),j^1(\alpha_{\tilde{\mathscr X}}(\eta))\right)=(\alpha_{\mathcal X}^{-1})^*\left[\tilde{\mathscr L}(j^1\phi,j^1\eta)\right]\,.
\ee

To sum up, the key point of the construction is that now the fixed background metric $g$ is no longer thought of as living on $\mathcal X$, but rather as a geometric object on the copy $\tilde{\mathcal X}$ in the fiber of the extended configuration bundle $\tilde{\mathcal Y}$. On the other hand, the pulled back metric variable $G=\eta^*g$ on $\mathcal X$ inherits a dynamical character from the covariance field $\eta$. The true dynamical fields of the parametrized theory are then given by $\phi$ and $\eta$. As discussed in \cite{lopezgotaymarsdenPFT,lopezgotayPFT}, the Euler-Lagrange equations for the fields $\phi$ remain unchanged while those for the covariance fields $\eta$ give the conservation of the stress-energy-momentum tensor. 
Such a construction provides us with a diffeomorphism-covariant formulation of field theories, where all fields are treated as variational entities  \cite{Anderson}\footnote{This has to be contrasted with considering the metric $g$ directly as a genuine field on $\mathcal X$, which will then make the Lagrangian $\textsf{Diff}(\mathcal X)$-equivariant but will prevent $g$ itself from being variational, unless one adds a source term (e.g., the Einstein-Hilbert Lagrangian) to the Lagrangian density.}.  

\subsection{Covariant Momentum Map}\label{PFTmultimomentmap}

Along the lines discussed in Sec. \ref{multisymplFT}, we shall now construct a covariant Hamiltonian formalism for (first order) parametrized field theories by considering the {covariant}, or parametrized, multiphase space $\tilde{\mathcal Z}\cong J^1(\tilde{\mathcal Y})^*$ equipped with a canonical Poincar\'e-Cartan $(n+1)$-form
\be
\tilde\Theta=\tilde p\,\dd^{n+1}x+p_A^\mu\dd y^A\wedge\dd^nx_\mu+\varrho_a^\mu\dd u^a\wedge\dd^nx_\mu\;.
\ee
In particular, as both $\phi$ and $\eta$ are dynamical variables for the extended theory, the covariant Hamiltonian $\tilde p$ and the multimomenta $p_A^\mu, \varrho_a^\mu$ (respectively conjugate to the multivelocities $v^A_\mu$ and $u^a_\mu$) are defined w.r.t. the Lagrangian \eqref{tildeL2} as follows
\be\label{PFTmultimomenta}
\tilde p=\tilde L-\frac{\partial\tilde L}{\partial v^A_\mu}v^A_\mu-\frac{\partial\tilde L}{\partial u^a_\mu}u^a_\mu,\qquad p_A^\mu=\frac{\partial\tilde L}{\partial v^A_\mu}=\frac{\partial L}{\partial v^A_\mu},\qquad\varrho_a^\mu=\frac{\partial\tilde L}{\partial u^a_\mu}=\mathcal T^{\mu\nu}u^b_\nu g_{ab}\;,
\ee
with $\mathcal T^{\mu\nu}=2\frac{\partial L}{\partial G_{\mu\nu}}$ the so-called \textit{Piola-Kirchhoff stress-energy-momentum tensor density} \cite{lopezgotaymarsdenPFT}.

\noindent
The multisymplectic $(n+2)$-form $\tilde\Omega=-\dd\tilde\Theta$ on $\tilde{\mathcal Z}$ then reads
\be
\tilde\Omega=\dd y^A \wedge \dd p_A^{\mu} \wedge \dd^n x_{\mu}+\dd u^a \wedge \dd \varrho_a^{\mu} \wedge \dd^n x_{\mu}-\dd\tilde p \wedge \dd^{n+1}x\;.
\ee

Let now $\mathcal G$ be a Lie group (perhaps infinite-dimensional) realizing the gauge group of the theory and denote by $\gi$ its Lie algebra. In the case of generally covariant field theories, $\mathcal G$ is a subgroup of $\textsf{Aut}(\tilde{\mathcal Y})$ covering diffeomorphisms on $\mathcal X$. Given an element $\xi \in \gi$, we denote by $\xi_\mathcal X,\xi_{\mathcal Y},\xi_{\tilde{\mathcal Y}}$, and $\xi_{\tilde{\mathcal Z}}$ the infinitesimal generators of the corresponding transformations on $\mathcal{X},\mathcal Y,\tilde{\mathcal Y}$, and $\tilde{\mathcal Z}$, i.e., the infinitesimal generators on $\mathcal{X},\mathcal Y,\tilde{\mathcal{Y}}$, and $\tilde{\mathcal Z}$ of the one-parameter group generated by $\xi$. The group $\mathcal{G}$ is said to act on $\tilde{\mathcal Z}$ by \textbf{covariant canonical transformation} if the $\mathcal G$-action corresponds to an infinitesimal multi-symplectomorphism, i.e.
\be
\mathcal L_{\xi_{\tilde{\mathcal Z}}}\tilde\Omega = 0\;, 
\ee
where $\mathcal L_{\xi_{\tilde{\mathcal Z}}}$ denotes the Lie derivative along $\xi_{\tilde{\mathcal Z}}$, while it is said to act by \textbf{special covariant canonical transformations} if $\tilde\Theta$ is $\mathcal G$-invariant, that is
\be
\mathcal L_{\xi_{\tilde{\mathcal Z}}}\tilde\Theta = 0\;. 
\ee 
This is the Hamiltonian counterpart of the $\mathcal G$-equivariance property \eqref{Gequiv} of the Lagrangian which in turn amounts to $\mathcal G$-invariance of the Lagrangian form $\Theta_{\tilde{\mathscr L}}$ on $J^1(\tilde{\mathcal Y})$ defined as the pull-back of $\tilde\Theta$ along the covariant Legendre transform.

On the extended covariant phase space, we now seek a general covariant macroscopic \emph{observable}, namely an object which is preserved by (stationary w.r.t.) the action of the spacetime diffeomorphisms and can be used to eventually characterise the fields configurations statistically on a mascoscopic, average scale. In general, group-theoretic terms, the object we need is the equivariant momentum map associated to the group action, which provides a generalization of the Hamiltonian function, generally comprising all conserved charges associated to the symplectic action of some dynamical group on a given phase space \cite{Souriau_b}. In this sense, the momentum map can be naturally used to generalize the Maxwell-Boltzmann-Gibbs approach to thermodynamics for the case of a vector-valued energy function, not just limited to time translation symmetry. As discussed in Sec. \ref{LGT}, Souriau's Lie Group Thermodynamics \cite{Souriau} first provides a remarkable \emph{group-covariant} generalization of Gibbs's equilibrium, with stationarity defined with respect to the action of the dynamical symmetry group of the system (e.g. Galilei, Poincar\'e) on its \emph{fully constrained} phase space (on-shell) \cite{Marle}. 
Similarly, in what follows, we will define a \emph{general-covariant} generalization of Gibbs's equilibrium as a radical conceptual extension of Souriau's Lie Group Thermodynamics for generally covariant field theories, where the covariant symmetry group of the system is \emph{gauge} and the symplectic phase space consists of the full \emph{extended phase space} of fields.

To this aim, let us notice that in analogy to the definition of momentum maps in symplectic geometry \cite{ham,marsdenratiu,Gui}, a {\bf covariant momentum map} (or a {\bf multimomentum map}) associated to the action of $\mathcal G$ on $\tilde{\mathcal Z}$ by covariant canonical transformations is a map
\be
\tilde{\mathcal J}: \tilde{\mathcal Z} \longrightarrow \gi^*\otimes \Lambda^n(\tilde{\mathcal Z})\;,
\ee
given by
\be\label{covmomentummap}
\dd\tilde{\mathcal{J}}(\xi) = \text{i}_{\xi_{\tilde{\mathcal Z}}}\tilde\Omega, 
\ee
where $\tilde{\mathcal J}(\xi)$ is the $n$-form on $\tilde{\mathcal Z}$ whose value at $\tilde{z} \in\tilde{\mathcal Z}$ is $\braket{\tilde{\mathcal J}(\tilde z),\xi}$ with $\braket{\cdot,\cdot}$ being the pairing between the Lie algebra $\mathfrak g$ and its dual $\mathfrak g^*$. 

Let then $\alpha\in\mathcal G$ be the transformation associated to $\xi\in\mathfrak g$, the covariant momentum map $\tilde{\mathcal J}$ is said to be $\text{Ad}^*$-\textbf{equivariant} if
\be
\tilde{\mathcal J}(\text{Ad}^{-1}_\alpha \xi)  = \alpha_{\tilde{\mathcal Z}}^*[\tilde{\mathcal J}(\xi)]\;.
\ee
When $\mathcal G$ acts by special covariant canonical transformations, the (special) covariant momentum map admits an explicit expression given by
\be\label{specialmmap}
\tilde{\mathcal J}(\xi) = \text{i}_{\xi_{\tilde{\mathcal Z}}}\tilde\Theta\;,
\ee
so that $\dd\tilde{\mathcal J}(\xi)= \dd\text{i}_{\xi_{\tilde{\mathcal Z}}}\tilde\Theta=(\mathcal L_{\xi_{\tilde{\mathcal Z}}} -\text{i}_{\xi_{\tilde{\mathcal Z}}} \dd)\tilde\Theta=\text{i}_{\xi_{\tilde{\mathcal Z}}}\tilde\Omega$. In particular -- and this is the case of interest for parametrized field theories -- if the $\mathcal G$-action on $\tilde{\mathcal Z}$ is the lift of an action of $\mathcal G$ on $\tilde{\mathcal Y}$, then for any $\xi\in\gi$ realized as a (complete) vector field $\xi_\mathcal X=\xi^\mu(x)\frac{\partial}{\partial x^\mu}$ on $\mathcal X$, we have
\be\label{xiYtilde}
\xi_{\tilde{\mathcal Y}}=\xi^\mu(x)\frac{\partial}{\partial x^\mu}+\xi^A(x,y,u,[\xi])\frac{\partial}{\partial y^A}\;,
\ee 
where in general $\xi^A(x,y,u,[\xi])$ is a smooth functional of $\xi^\nu$ which for tensor field theories reads
\be
\xi^A(x,y,u,[\xi])=C^{A\rho}_\nu(x,y,u)\xi^\nu_{,\rho}(x)+C^A_\nu(x,y,u)\xi^\nu(x)\;,
\ee
with coefficients $C$ depending on $x^\mu$, $y^B$ and $u^b$, and reduces to $\xi^A(x,y,u,[\xi])=C^A_\nu(x,y,u)\xi^\nu(x)$ in the case of a scalar field. Note that in Eq. \eqref{xiYtilde} there are no components in the ``$u$-directions'' since, as discussed in the previous section (cfr. Eq. \eqref{trivext}), for any $\alpha_\mathcal X\in\textsf{Diff}(\mathcal X)$, there is a lift $\alpha_\mathcal Y\in\textsf{Aut}(\mathcal Y)$ which is trivially extended to $\tilde{\mathcal Y}$ by $\alpha_{\tilde{\mathcal Y}}:(y,u,x)\mapsto(\alpha_\mathcal Y(y), u,\alpha_\mathcal X(x))$.
In coordinates, Eq. \eqref{specialmmap} then reads
\be\label{specialmmap2}
\braket{\tilde{\mathcal J}(\tilde z),\xi}=\left(\tilde p\, \xi^{\mu}+p_A^{\mu} \xi^A\right)\dd^n x_{\mu} -p_A^{\mu} \xi^{\nu} \dd y^A  \wedge \dd^{n-1} x_{\mu\nu}-\varrho_a^{\mu} \xi^{\nu} \dd u^a  \wedge \dd^{n-1} x_{\mu\nu}, 
 \ee
where $\dd^{n-1} x_{\mu \nu} = \text{i}_{\partial_{\nu}}\text{i}_{\partial_{\mu}} \dd^{n+1}x$.

\section{Covariant Momentum Map as a Statistical Observable}\label{canonicalphasespaceEMmap}

The covariant momentum map associated to the lifted action to $\tilde{\mathcal Z}$ of the automorphisms of $\tilde{\mathcal Y}$ covering diffeomorphisms on the base manifold $\mathcal X$ provides us with a general-covariant extension of the notion of ``Hamiltonian''. However, formally, we are still one step behind our goal to extend the generalized notion of thermodynamic equilibrium states \'a la Souriau to covariant field theories. Indeed, a variational Maximum Entropy derivation of the generalized Gibbs states requires to constrain the mean value of the momentum map. However, the momentum map $\tilde{\mathcal J}:\tilde{\mathcal Z}\to\mathfrak g^*\otimes\Lambda^{n}(\tilde{\mathcal Z})$ has a differential form component and the definition of its mean value is simply not defined. In facts, as an $n$-form on $\tilde{\mathcal Z}$, $\tilde{\mathcal J}$ should be integrated over co-dimension 1 hypersurfaces in a $(n+1)$-dimensional submanifold of $\tilde{\mathcal Z}$ in order to produce a proper functional of the fields. How shall we chose such hypersurfaces? 

What we need is to introduce a \emph{foliation} of spacetime and consequently of the bundles over it. Following the approach of the canonical formulation of field theory, we shall decompose the base manifold $\mathcal X$ into a smooth disjoint union of space-like hypersurfaces $\Sigma$\footnote{In what follows, we assume the spacetime manifold to be globally hyperbolic, i.e., $\mathcal X\cong\Sigma\times\mathbb R$, so that the foliation introduced to construct the canonical formalism covers the whole manifold thus avoiding technicalities concerning the possibility of defining the canonical formalism only locally.}, make such slicing covariant and compatible with the Hamiltonian action of the diffeomorphisms and eventually integrate our momentum map over sections restricted on the slices.\footnote{An alternative approach to get an observable momentum map functional is to use a covariant phase space approach \`a la Helein (see e.g. \cite{helein1,helein3,helein4} and references therein for details as well as its relation with the multi-symplectic formulation). This consists in introducing a slicing directly at the level of the fields target space rather than considering the pulled-back quantities along sections on co-dimension one hypersurfaces in the base manifold.}

Let then $\Sigma$ be a compact, oriented, connected, boundaryless 3-manifold and let $\textsf{Emb}_G(\Sigma, \mathcal X)$ be the set of all space-like embeddings of $\Sigma$ in $\mathcal X$. A foliation $\mathfrak{s}_\mathcal X$ of $\mathcal X$ then corresponds to a 1-parameter family $\lambda\mapsto\tau_\lambda$ of space-like embeddings $\tau_\lambda\in\textsf{Emb}_G(\Sigma, \mathcal X)$ of $\Sigma$ in $\mathcal X$, i.e.
\be\label{fol}
\mathfrak{s}_\mathcal X:\Sigma\times\mathbb R\to\mathcal X\quad\text{by}\quad (\vec x,\lambda)\mapsto \mathfrak{s}_\mathcal X(\vec x,\lambda)\,,
\ee
such that
\be
\tau\equiv\tau_\lambda: \Sigma\to\mathcal X\quad\text{by}\quad \tau(\vec x)\equiv\tau_\lambda(\vec x):=\mathfrak{s}_\mathcal X(\vec x,\lambda)\;,
\ee
where $\vec x$ is a shorthand notation for the spatial coordinates $x^i$, $i=1,\dots,n$, on the space-like hypersurface $\Sigma_\tau=\tau(\Sigma)$. The generator of $\mathfrak s_\mathcal X$ is a complete vector field $\zeta_\mathcal X$ on $\mathcal X$ everywhere transverse to the slices defined by
\be
\dot\tau(\vec x)=\frac{\partial}{\partial\lambda}\mathfrak{s}_\mathcal X(\vec x,\lambda)=\zeta_\mathcal X\left(\mathfrak{s}_\mathcal X(\vec x,\lambda)\right)\,.
\ee
A foliation of $\mathcal X$ induces a \textit{compatible slicing} of bundles over it whose generating vector fields project onto $\zeta_\mathcal X$. The flow of such a generating vector field defines a one-parameter group of bundle automorphisms. For parametrized field theories we are interested in a so-called \textbf{$\mathcal G$-slicing} in which case the one-parameter group of automorphisms of the extended configuration bundle is induced by a one-parameter subgroup of the gauge group $\mathcal G$, i.e., $\zeta_{\tilde{\mathcal Y}}=\xi_{\tilde{\mathcal Y}}$ for some $\xi\in\mathfrak{g}$. The corresponding slicing $\mathfrak{s}_{\tilde{\mathcal Z}}$ of $\tilde{\mathcal Z}$ is then generated by the canonical lift $\zeta_{\tilde{\mathcal Z}}=\xi_{\tilde{\mathcal Z}}$ of $\xi_{\tilde{\mathcal Y}}$ to $\tilde{\mathcal Z}$ whose flow defines a one-parameter group of bundle automorphisms by special canonical transformations on $\tilde{\mathcal Z}$ (i.e., $\mathcal L_{\xi_{\tilde{\mathcal Z}}}\tilde\Theta=0$)\footnote{As already stressed before, this essentially reflects the equivariance property of the Lagrangian density w.r.t. to the one-parameter groups of automorphisms associated to the induced slicings of $J^1(\tilde{\mathcal Y})$ and $\Lambda^{n+1}(\mathcal X)$.}.

Spatial fields will then be identified with smooth sections of the pull-back bundle $\mathcal Y_\tau\to\Sigma_\tau$ over a Cauchy surface given by $\varphi:=\phi_\tau=\tau^*\phi$. Note that, as the subscript $\tau$ is meant to recall, the spatial fields $\varphi(\vec x)=\phi_\tau(\vec x)$ are functionals of the embedding $\tau$ and at the same time functions of the point $\vec x$ on the spatial slice. Moreover, according to the parametrization procedure discussed in Sec. \ref{PFTcovfields}, the space-like embedding $\tau\in\textsf{Emb}_G(\Sigma, \mathcal X)$ acquires a dynamical character through the covariance fields $\eta$. Indeed, we have $\tau=\eta^{-1}\comp\tilde\tau$ for a given space-like embedding $\tilde\tau\in\textsf{Emb}_g(\Sigma, \tilde{\mathcal X})$ of $\Sigma$ into $\tilde{\mathcal X}$ associated to the slicing of $\tilde{\mathcal X}$ w.r.t. the fixed metric $g$. The canonical parametrized configuration space then consists of the pairs $(\varphi, \tau)$ of spatial fields defined over a Cauchy slice and the space-like embeddings identifying a $\mathcal G$-slicing of spacetime w.r.t. one-parameter subgroups of diffeomorphisms. Let then $(x^0,x^1,\dots,x^n)$ be a chart on $\mathcal X$ adapted to $\tau$, i.e. such that $\Sigma_\tau$ is locally a level set of $x^0$. Denoting by $(\varphi,\Pi,\tau,P)$ a point in the canonical parametrized phase space $T^*\tilde{\mathscr{Y}_\tau}=T^*\mathscr{Y}_\tau\times T^*\textsf{Emb}_G(\Sigma, \mathcal X)$, the canonical symplectic structure $\tilde{\omega}_\tau$ on $T^*\tilde{\mathscr{Y}_\tau}$ reads as \cite{ishamsymplgeom}
\be\label{fieldsymplecticstructure}
\tilde{\omega}_\tau(\varphi,\Pi,\tau,P)=\int_{\Sigma_\tau}\left(\dd\varphi^A\wedge\dd\Pi_A+\dd\tau^\mu\wedge\dd P_\mu\right)\otimes\dd^n x_0\;.
\ee

The canonical and multi-symplectic parametrized phase space descriptions are related as follows. Following the construction of \cite{gimmsyII} (cfr. Ch. 5), the multisymplectic structure on $\tilde{\mathcal Z}$ induces a presymplectic structure on the space $\tilde{\mathscr Z}_\tau$ of sections of the bundle $\tilde{\mathcal Z}_\tau\to\Sigma_\tau$ given by
\be\label{PFTpresympl}
\tilde\Omega_\tau(\sigma)(V,W)=\int_{\Sigma_\tau}\sigma^*(\text{i}_W\text{i}_V\tilde\Omega)\qquad,\qquad\sigma\in\tilde{\mathscr Z}_\tau\;,\; V,W\in T_\sigma\tilde{\mathscr{Z}}_\tau
\ee
which in turn is related to $\tilde{\omega}_\tau$ via $\tilde\Omega_\tau=R_\tau^*\tilde\omega_\tau$, where $R_\tau$ is the bundle map $R_\tau:\tilde{\mathscr{Z}_\tau}\to T^*\tilde{\mathscr{Y}_\tau}$ relating in adapted coordinates the momenta $\Pi_A$ and $P_a$ respectively to the temporal components of the multimomenta $p_A^\mu$ and $\varrho_a^\mu$ as
\be\label{canonmomenta}
\Pi_A=p_A^0\comp\sigma\qquad,\qquad P_a=\varrho_a^0\comp\sigma\;.
\ee
In particular \cite{gimmsyII}, $\ker T_\sigma R_\tau=\ker\tilde{\Omega}_\tau(\sigma)$ and the canonical parametrized phase space $T^*\tilde{\mathscr{Y}_\tau}$ is thus isomorphic to the quotient $\tilde{\mathscr{Z}}_\tau/\ker\tilde{\Omega}_\tau$.


Let now $\sigma\in\tilde{\mathscr Z}\equiv\Gamma(\mathcal X,\tilde{\mathcal Z})$ be a section of the bundle $\tilde{\mathcal Z}$ over $\mathcal X$, and let $\alpha_{\tilde{\mathcal Z}}:\tilde{\mathcal Z}\to\tilde{\mathcal Z}$ be a covariant canonical transformation covering a diffeomorphism $\alpha_\mathcal X:\mathcal X\to\mathcal X$ whose induced action on sections is given by $\alpha_{\tilde{\mathscr Z}}(\sigma)=\alpha_{\tilde{\mathcal Z}}\comp\sigma\comp\alpha_\mathcal X^{-1}$ (cfr. Eq. \eqref{actiononsections2}). The corresponding transformation on $\tilde{\mathscr{Z}}_\tau\equiv\Gamma(\Sigma_\tau,\tilde{\mathcal Z})$ given by
\begin{align}
\alpha_{\tilde{\mathscr Z}_\tau}:\tilde{\mathscr Z}_{\eta^{-1}\comp\tilde\tau}&\longrightarrow\tilde{\mathscr Z}_{\alpha_{\tilde{\mathscr{X}}}(\eta)^{-1}\comp\tilde\tau}\nonumber\\
\sigma&\longmapsto\alpha_{\tilde{\mathscr Z}_\tau}(\sigma)=\alpha_{\tilde{\mathcal Z}}\comp\sigma\comp\alpha_\tau^{-1}\;,\label{cansecaction}
\end{align}
with $\alpha_{\tilde{\mathscr{X}}}(\eta)$ defined in \eqref{actiononsections3} and $\alpha_\tau:=\alpha_\mathcal X|_{\Sigma_\tau}$ is a (special) covariant canonical transformation relative to the presymplectic 2-form \eqref{PFTpresympl} if $\alpha_{\tilde{\mathcal Z}}$ is a (special) covariant canonical transformation \cite{gimmsyII}, and we have the following:

\begin{theorem}
The covariant multimomentum map \eqref{covmomentummap} associated to the $\mathcal G$-action on $\mathcal{\tilde{Z}}$ induces a (\textbf{parametrized}) \textbf{energy-momentum map}, $\mathcal{\tilde{E}}_{\tau}: \mathscr{\tilde{Z}}_{\tau}\longrightarrow\mathfrak g^*$, defined by
\be\label{parametrizedemm}
\mathcal{\tilde{E}}_{\tau}(\sigma,\eta)=\mathcal{\tilde{E}}_{\eta^{-1}\comp\tilde\tau}(\sigma):=\int_{\Sigma_\tau}\sigma^*\braket{\tilde{\mathcal J},\xi}\;.
\ee
Such a parametrized energy-momentum map is $\text{Ad}^*$-equivariant w.r.t. the action \eqref{cansecaction}, i.e.
\be\label{parametrizedenergymommapequivariance}
\braket{\tilde{\mathcal{E}}_\tau(\sigma,\eta),\text{Ad}^{-1}_{\alpha}\xi}=\braket{\alpha_{\tilde{\mathscr Z}_\tau}^*\bigl[\tilde{\mathcal{E}}_\tau(\sigma,\eta)\bigr],\xi}\;.
\ee
\end{theorem}
\begin{proof}By explicitly writing the r.h.s. of Eq. \eqref{parametrizedenergymommapequivariance}, we get:
\begin{align}
\braket{\alpha_{\tilde{\mathscr Z}_\tau}^*\bigl[\tilde{\mathcal{E}}_\tau(\sigma,\eta)\bigr],\xi}&=\braket{\tilde{\mathcal{E}}_{\alpha_{\tilde{\mathscr{X}}}(\eta)^{-1}\comp\tilde\tau}(\alpha_{\tilde{\mathscr Z}_{\tau}}(\sigma))}\nonumber\\
&=\int_{(\eta\comp\alpha^{-1}_{\mathcal X})^{-1}\comp\tilde\tau(\Sigma)}\,(\alpha_{\tilde{\mathcal Z}}\comp\sigma\comp\alpha_\tau^{-1})^*\braket{\tilde{\mathcal J},\xi}\nonumber\\
&=\int_{\alpha_{\mathcal X}\comp(\eta^{-1}\comp\tilde\tau)(\Sigma)}\,(\alpha_\tau^{-1})^*\,\sigma^*\,\alpha_{\tilde{\mathcal Z}}^*\braket{\tilde{\mathcal J},\xi}\nonumber\\
&=\int_{\eta^{-1}\comp\tilde\tau(\Sigma)}\,\sigma^*\,\alpha_{\tilde{\mathcal Z}}^*\braket{\tilde{\mathcal J},\xi}\quad\,\qquad(\text{change of variables})\nonumber\\
&=\int_{\eta^{-1}\comp\tilde\tau(\Sigma)}\,\sigma^*\,\braket{\tilde{\mathcal J},\text{Ad}^{-1}_\alpha\xi}\qquad\,(\text{Ad$^*$-equivariance of $\tilde{\mathcal J}$})\nonumber\\
&=\braket{\tilde{\mathcal{E}}_{\eta^{-1}\comp\tilde\tau}(\sigma),\text{Ad}^{-1}_{\alpha}\xi}\nonumber\\
&=\braket{\tilde{\mathcal{E}}_\tau(\sigma,\eta),\text{Ad}^{-1}_{\alpha}\xi}\label{PEMMequivariance}\;,
\end{align}
which proves the desired equivariance property.
\end{proof}

\subsection{Energy-Momentum Map and Initial Value Constraints}

The energy-momentum map is intimately related to the initial value constraints which generate the covariant gauge freedom thus providing on the one hand a fundamental link between the dynamics and the gauge group, and on the other hand encoding in a single geometrical object all the physically relevant information about a given classical field theory \cite{gimmsyII}. 

To see this, let us denote by $\tilde{\mathscr{P}}_\tau$ the primary constraint submanifold in $T^*\tilde{\mathscr{Y}}_\tau$ defined as $\tilde{\mathscr{P}}_\tau=R_\tau(\tilde{\mathscr{N}}_\tau)\subset T^*\tilde{\mathscr{Y}}_\tau$ with $\tilde{\mathscr{N}}_\tau=\mathbb F\mathscr{L}((j^1\tilde{\mathscr{Y}})_\tau)\subset\tilde{\mathscr{Z}}_\tau$, $\mathbb F\mathscr{L}$ being the Legendre transform. For lifted actions -- and this is the case for a $\mathcal G$-slicing discussed before -- the projection $\tilde{\mathcal J}_H:\tilde{\mathscr{P}_\tau}\to\mathfrak g^*$ on $\tilde{\mathscr{P}_\tau}$ of the parametrized energy-momentum map \eqref{parametrizedemm} encodes the first class secondary constraints in its transverse and tangential components to the spatial slice. Indeed at the level of densities, using adapted coordinates and recalling the expressions \eqref{specialmmap}, \eqref{specialmmap2}, we have
\begin{eqnarray}\label{emmapdensity1}
\sigma^*(\text{i}_{\zeta_{\tilde{\mathcal Z}}}\tilde\Theta)=\Bigl[(p_A^0\comp\sigma)(\zeta^A\comp\sigma&-&\zeta^\mu\sigma^A_{,\mu})-(\varrho_a^0\comp\sigma)\zeta^\mu\sigma^a_{,\mu}\\ \nonumber
&+&\Bigl(\tilde p\comp\sigma+(p_A^\mu\comp\sigma)\sigma^A_{,\mu}+(\varrho_a^\mu\comp\sigma)\sigma^a_{,\mu}\Bigr)\zeta^0\Bigr]\dd^nx_0
\end{eqnarray}
for any $\sigma\in\tilde{\mathscr{Z}}_\tau$ and $\zeta_{\tilde{\mathcal Z}}$ the canonical lift to $\tilde{\mathcal Z}$ of the generator of the $\mathcal G$-slicing $\zeta_{\mathcal X}=\xi_{\mathcal X}$. Now, for any $\sigma$ \textit{holonomic lift} of $(\varphi,\Pi,\tau,P)$ to $\tilde{\mathscr{N}}_\tau$, that is $\sigma\in R_\tau^{-1}\{(\phi,\Pi,\tau,P)\}\cap\tilde{\mathscr{N}}_\tau$, we have $\sigma^A=\phi^A\bigl|_{\Sigma_\tau}=\varphi^A$ and $\sigma^a=\eta^a\bigl|_{\Sigma_\tau}$ so that
\be
\zeta^A\comp\sigma-\zeta^\mu\sigma^A_{,\mu}=(\zeta^A\comp\phi-\zeta^\mu\phi^A_{,\mu})\bigl|_{\Sigma_\tau}=-\left(\mathcal L_{\zeta}\phi\right)^A\bigl|_{\Sigma_\tau}=:-\dot\varphi^A\;,
\ee
and
\be
\tilde p\comp\sigma+(p_A^\mu\comp\sigma)\sigma^A_{,\mu}+(\varrho_a^\mu\comp\sigma)\sigma^a_{,\mu}=\tilde{L}(\sigma)\;,
\ee
where we used the expressions for the covariant Hamiltonian and the multimomenta (Eqs. \eqref{PFTmultimomenta}) and for the canonical momenta (Eq. \eqref{canonmomenta}). Hence, Eq. \eqref{emmapdensity1} yields
\be\label{emmapdensity2}
\sigma^*(\text{i}_{\zeta_{\tilde{\mathcal Z}}}\tilde\Theta)=-\left(\Pi_A\dot\varphi^A+P_a\zeta^\mu\eta^a_{,\mu}-\tilde{L}(\sigma)\zeta^0\right)\dd^nx_0\;,
\ee
from which, by using the fact that 
$$\tilde{L}(\sigma)\zeta^0\dd^nx_0=\tau^*\textit{i}_{\zeta_{\mathcal X}}\tilde{\mathscr L}(j^1\phi,j^1\eta)=\textit{i}_{\zeta_{\mathcal X}}\tilde{\mathscr L}(j^1\varphi,\dot\varphi,j^1\eta_\tau,\dot\eta_\tau),$$ 
it follows that the parametrized energy-momentum map \eqref{parametrizedemm} induces a functional on $\tilde{\mathscr{P}}_\tau$
\be
\tilde{\mathcal J}_H:\tilde{\mathscr{P}}_\tau\longmapsto\mathfrak g^*\;,
\ee
given by
\begin{eqnarray}\label{diffeqmommap}
\braket{\tilde{\mathcal J}_H(\varphi,\Pi,\tau,P),\zeta}&=&-\int_{\Sigma_\tau}\dd^nx_0(\zeta^\mu\mathcal H^{(\varphi)}_\mu+\zeta^\mu P_\mu)\\ \nonumber
&=&-\left(H^{(\varphi)}(\zeta)(\varphi,\Pi,\tau)+P(\zeta)(\tau,P)\right)\;,
\end{eqnarray}
where $P_\mu=\eta^a_{,\mu}P_a$ is the pull-back of $P_a$ to $\Sigma$ along $\eta$. The functional \eqref{diffeqmommap} on $\tilde{\mathscr{P}}_\tau$ is nothing but the total Hamiltonian, whose components in the tangential and transversal direction to the spatial slice yield the super-momenta and Hamiltonian constraints, respectively. 

As we will discuss in the next section, the moment map in \eqref{diffeqmommap} is equivariant w.r.t. the $\mathcal G$-action on the canonical parametrized phase space and it provides us with a representation of the Lie algebra of $\textsf{Diff}(\mathcal X)$ on the parametrized phase space. Differently from the standard case of instantaneous canonical formalism where the equivariance of the corresponding functional is spoiled by the algebra of constraints, in our derivation the key ingredient which leads to a Lie algebra (anti)homomorphism between the algebra $\textsf{diff}(\mathcal X)$ and the Poisson bracket algebra of observable functionals on the parametrized phase space relies on the introduction of the covariance fields. Indeed, as we will discuss in the next section, the presence of the covariance fields in turn allows to induce a corresponding $\mathcal G$-action on the space of embeddings of $\Sigma$ into $\mathcal X$. The action of the diffeomorphism group on the embeddings $\tau$ induced via the action on $\eta$ is what ensure the equivariance of the parametrized energy-momentum map (cfr. Eq. \eqref{PEMMequivariance}).

\subsection{Representation of Spacetime Diffeomorphisms: $\textsf{Diff}(\mathcal X)$-equivariant Momentum Map}\label{diffreprmmap}

Let $\mathcal G=\textsf{Diff}(\mathcal X)$ be the group of diffeomorphisms (i.e., smooth and invertible active point transformations\footnote{This should not be confused with the pseudo-group of passive transformations which describes the relations between overlapping pairs of coordinate charts. For details on the notion of passive and active diffeomorphisms and their connection we refer to \cite{RovelliDiffeo} and references therein.}) of the spacetime manifold $\mathcal X$. The Lie algebra $\mathfrak g=\textsf{diff}(\mathcal X)$ can be realized as the set of all (complete) vector fields on $\mathcal X$. Indeed, to any element $\xi\in\textsf{diff}(\mathcal X)$, we can associate a vector field $\xi_\mathcal X\in\mathfrak X(\mathcal X)$ generating a one-parameter group $\alpha_\lambda^{\xi}$ of spacetime diffeomorphisms by
\be
\frac{\dd}{\dd\lambda}\alpha_\lambda^{\xi}(x)=\xi_\mathcal X(\alpha_\lambda^{\xi}(x))\qquad\forall\,x\in\mathcal X\;.
\ee 
In the standard setting of the instantaneous canonical formalism, the generating vector field is decomposed into its ``lapse'' and ``shift'' components, i.e. it can be written as the sum of two vectors which are respectively normal and tangent vectors to the space-like hypersurface, say
\be\label{paralperpdecomposition1}
\xi_\mathcal X^\mu(\vec x)=N(\vec x)n^\mu(\tau(\vec x))+N^k(\vec x)\tau^\mu_{,k}(\vec x)
\ee
where $n^\mu=G^{\mu\nu}n_\nu$ is the future-pointing normal such that $\tau^*n=0$ for any $\vec x\in\Sigma_\tau$ and $G^{\mu\nu}n_\nu n_\nu=-1$, $N\in C^\infty(\Sigma,\mathcal X)$ is the \textit{lapse function}, and $\vec N\in T\Sigma$ the \textit{shift vector} of the foliation \cite{Giulini, ADM}. Lapse and shift respectively give the magnitude of the normal and tangential deformation at every point on a spatial hypersurface, specifying how the neighbouring slice is embedded in spacetime, and play the role of arbitrary Lagrange multipliers in the action implementing the first-class constraints of the theory.

The decomposition into perpendicular and tangential directions to the spatial hypersurfaces deforms the algebra $\textsf{diff}(\mathcal X)$. Indeed, the projected constraint functions do not form a genuine Lie algebra but rather a Lie algebroid structure known as the \textit{Dirac hypersurface-deformation algebra} \cite{Diracalgebra1,Diracalgebra2,hypersurfacedef}. After smearing with lapse and shift, the algebra reads\footnote{We refer to \cite{ishamkuchar1} Sec. 3.3 for the detailed calculation of the constraint algebra for the case of a canonical parametrized scalar field theory.}
\be\label{DHDAconstraints}
\begin{aligned}
&\{H[\vec N],H[\vec M]\}=H[\mathcal L_{\vec N}\vec M]\,,\\
&\{H[\vec N],H[M]\}=H[\mathcal L_{\vec N}M]\,,\\
&\{H[N],H[M]\}=H[\vec K]\,,
\end{aligned}
\ee
where $\mathcal L_{\vec N}$ is the Lie derivative along the vector field $\vec N$ (i.e., $\mathcal L_{\vec N}\vec M=[\vec N,\vec M]$ and $\mathcal L_{\vec N}M=N^kM_{,k}$), and $\vec K$ is such that
\be
K^j=Q^{jk}\left(NM_{,k}-MN_{,k}\right)\;.
\ee

Therefore, the Poisson brackets between two super-Hamiltonians $H[N]$, $H[M]$ do not only depend on the pair of shift functions $N$ and $M$, but also explicitly on the canonical embedding variable through the induced (inverse) 3-metric $Q^{ab}$. This implies that the full algebra of spacetime diffeomorphisms cannot be represented on the canonical parametrized phase space as it is not homomorphically mapped into a Poisson bracket algebra on the parametrized phase space. Only the subalgebra $\textsf{diff}(\Sigma)$ of spatial diffeomorphisms can be represented in the canonical formalism, as the map $\vec N\mapsto H[\vec N]$ provides a Lie algebra homomorphism of $\textsf{diff}(\Sigma)$ into a Poisson bracket algebra on the phase space of the system according to the first equation in \eqref{DHDAconstraints}. 

This reflects into the fact that the (\textit{instantaneous}) \textit{energy-momentum map}
\be
\mathcal E_\tau: T^*\textsf{Emb}_G(\Sigma,\mathcal X)\times T^*\mathscr Y_\tau\to\Lambda^0_d\times\Lambda^1_d\,,
\ee
given by
\be\label{instenergymomentummap}
\begin{aligned}
\mathcal E_\tau[N,\vec N]&=\int_\Sigma\dd^3x\braket{(N,\vec N),\mathcal E_\tau}\\
&=\int_\Sigma\dd^nx_0\,\left(N\mathcal H+N^k\mathcal H_k\right)=H[N]+H[\vec N]
\end{aligned}
\ee
with $\Lambda^0_d\times\Lambda^1_d$ the dual of the space of lapses and shifts\footnote{$\Lambda^0_d$ and $\Lambda^1_d$ respectively denote the spaces of function densities and 1-form densities on $\Sigma$.}, is not a true momentum map. Indeed, as it can be checked by direct computation, from the Dirac algebra \eqref{DHDAconstraints} it follows that
\be
\{\mathcal E_\tau[N,\vec N],\mathcal E_\tau[M,\vec M]\}=\mathcal E_\tau[\mathcal L_{\vec N}M-\mathcal L_{\vec M}N,\mathcal L_{\vec N}\vec M+\vec K]\,,
\ee
i.e., $\mathcal E_\tau$ is not (infinitesimally) equivariant.
Only if we restrict to the spatial diffeomorphisms, i.e., the subgroup $\mathcal G_\tau=\textsf{Diff}(\Sigma_\tau)$ of transformations which stabilize the image of $\tau$, the energy-momentum map \eqref{instenergymomentummap} induces a momentum map w.r.t. the $\mathcal G_\tau$-action, say
\be\label{spatialmomentummap1}
J_\tau:=\mathcal E_\tau\bigl|_{\mathfrak{g}_\tau}:T^*\textsf{Emb}_G(\Sigma,\mathcal X)\times T^*\mathscr Y_\tau\to\mathfrak{g}_\tau^*
\ee
such that
\be\label{spacediffmommap}
J_\tau[\vec N]=\int_\Sigma\dd^nx_0\,N^k\mathcal H_k=H[\vec N]\;,
\ee
which, as expected from the first of equations \eqref{DHDAconstraints}, is equivariant under the $\mathcal G_\tau$-action, thus providing a representation of the Lie algebra $\mathfrak{g}_\tau=\textsf{diff}(\Sigma_\tau)$ of spatial diffeomorphisms in terms of a Poisson bracket algebra of functionals on the extended phase space.

Nevertheless, the action of the full group of diffeomorphisms on the parametrized phase space of the theory can be recovered by taking into account the action of $\textsf{Diff}(\mathcal X)$ on the embedding fields. Specifically, the left action of $\textsf{Diff}(\mathcal X)$ on $\mathcal X$ induces a natural left action of $\textsf{Diff}(\mathcal X)$ on the space $\textsf{Emb}(\Sigma,\mathcal X)$ of all embeddings of $\Sigma$ in $\mathcal X$
\be\label{GactionEmb}
\Psi:\textsf{Emb}(\Sigma,\mathcal X)\times\textsf{Diff}(\mathcal X)\longrightarrow\textsf{Emb}(\Sigma,\mathcal X)\qquad\text{by}\qquad(\tau,\alpha_\mathcal X)\longmapsto \alpha_\mathcal X \comp \tau\;,
\ee
which carries the points $\tau(\vec x,\lambda)$ of the hypersurface $\Sigma_\tau=\tau(\Sigma)$ into new spacetime positions $\alpha_\mathcal X(\tau(\vec x,\lambda))$ forming a new hypersurface. Indeed, according to the action \eqref{actiononsections3} of the diffeomorphism group on $\eta$, we have
\begin{align}
\alpha\cdot\tau&=\alpha_{\tilde{\mathscr X}}(\eta)^{-1}\comp\tilde\tau\nonumber\\
&=\left(\eta\comp\alpha_\mathcal X^{-1}\right)^{-1}\comp\tilde\tau\nonumber\\
&=\alpha_{\mathcal X}\comp\eta^{-1}\comp\tilde\tau\nonumber\\
&=\alpha_{\mathcal X}\comp\tau\;.
\end{align}
The corresponding generating vector field, 
\be\label{Embvectorfield}
\xi_\tau(\vec x)=\xi_{\mathcal X}(\tau(\vec x))=\xi_\mathcal X^\mu(\tau(\vec x))\frac{\partial}{\partial x^\mu}\biggl|_{\tau(\vec x,\lambda)}\,,
\ee
yields a representation of the algebra $\textsf{diff}(\mathcal X)$ by vector fields on $\textsf{Emb}_G(\Sigma,\mathcal X)$. 

Following \cite{ishamkuchar1} and the  references therein, we think of the set $\textsf{Emb}_G(\Sigma,\mathcal X)$ of space-like embeddings of $\Sigma$ into $\mathcal X$ as an infinite-dimensional manifold.\footnote{Indeed, $\textsf{Emb}_G(\Sigma,\mathcal X)$ is a open subset of the set $\textsf{Emb}(\Sigma,\mathcal X)$ of all embeddings (not necessarily space-like) of $\Sigma$ into $\mathcal X$ which in turn is an open subset of the infinite-dimensional manifold $C^\infty(\Sigma,\mathcal X)$ of smooth functions from $\Sigma$ into $\mathcal X$ equipped with the compact-open topology, thus inheriting its differential structure.} The tangent space $T_{\tau}\textsf{Emb}_G(\Sigma,\mathcal X)$ at $\tau\in\textsf{Emb}_G(\Sigma,\mathcal X)$ is then defined as
$$
T_{\tau}\textsf{Emb}_G(\Sigma,\mathcal X):=\{\xi_\tau:\Sigma\to T\mathcal X\;|\;\xi_\tau(\vec x)\in T_{\tau(\vec x)}\mathcal X\;,\;\forall\,\vec x\in\Sigma\},
$$
and similarly
$$
T^*_{\tau}\textsf{Emb}_G(\Sigma,\mathcal X):=\{\gamma_\tau:\Sigma\to T^*\mathcal X\;|\;\gamma_\tau(\vec x)\in T^*_{\tau(\vec x)}\mathcal X\;,\;\forall\,\vec x\in\Sigma\},
$$
with $L^2$-dual pairing given by
\be\label{L2pairing}
\braket{\gamma_\tau,\xi_\tau}:=\int_\Sigma\dd^nx_0\sqrt{\det Q(\vec x)}\,\gamma_{\mu}(\tau(\vec x))\xi^\mu(\tau(\vec x))\;,
\ee
where $Q=\tau^*G$ is the induced metric on $\Sigma$. 

Although a generic diffeomorphism $\alpha\in\textsf{Diff}(\mathcal X)$ in general would not preserve the space-like nature of the embeddings, for any $\tau\in\textsf{Emb}_G(\Sigma,\mathcal X)$ there exists a open neighborhood of the identity in $\textsf{Diff}(\mathcal X)$ such that the transformed embedding $\alpha\comp\tau$ is still space-like. An element $\xi\in\textsf{diff}(\mathcal X)$, realized as a complete vector field $\xi_\mathcal X$ on $\mathcal X$, thus yields a vector field on $\textsf{Emb}_G(\Sigma,\mathcal X)$ by the prescription \eqref{Embvectorfield}. Such a vector field restricted to the embeddings can be then decomposed into the corresponding lapse and shift components which now become functionals of $\tau$. 

By taking this dependence into account, for any $\xi\in\textsf{diff}(\mathcal X)$ we can define a new Hamiltonian functional on the parametrized phase space, which is related to the equivariant momentum map \eqref{diffeqmommap} via
\begin{align}\label{equivnewham}
H(\xi)(\varphi,\Pi,\tau,P):&=-\braket{\xi,\tilde{\mathcal J}_H(\varphi,\Pi,\tau,P)}\nonumber\\
&=\int_\Sigma\dd^nx_0\,\xi_\mathcal X^\mu(\tau(\vec x))\mathcal H_\mu(\varphi,\Pi,\tau,P)\nonumber\\
&=\int_\Sigma\dd^nx_0\,\xi_\mathcal X^\mu(\tau(\vec x))\Bigl(\mathcal H^{(\varphi)}_\mu+P_\mu\Bigr)\,.
\end{align}

This result is compatible with the procedure introduced by Isham and Kucha\v{r} in \cite{ishamkuchar1} for the case of a parametrized scalar field theory.\footnote{In this case, the key step in representing spacetime diffeomorphisms was the observation that the embedding variables -- which tell how the model Cauchy surface $\Sigma$ lies in the spacetime manifold -- provide a link between the spatial and spatiotemporal pictures.}

At the infinitesimal level, the equivariance property of the Hamiltonian \eqref{equivnewham} can be seen as follows. We use the fact that the Poisson bracket between the unprojected constraint functions vanishes strongly and the canonical Poisson brackets $\{\tau^\mu(\vec x),P_\nu(\vec x^{\,\prime})\}=\delta^\mu_\nu\delta(\vec x,\vec x^{\,\prime})$ between the embeddings and their conjugate momenta as well as that the embeddings commute with $\mathcal H^{(\varphi)}$. Then, for any two Lie algebra elements $\xi,\zeta\in\textsf{diff}(\mathcal X)$, corresponding to vector fields $\xi_\mathcal X(\tau(\vec x)),\zeta_\mathcal X(\tau(\vec x))$ generating one-parameter groups of spacetime diffeomorphisms, we have

\begin{align}
\{H(\xi),H(\zeta)\}&=-\int_\Sigma\dd^nx_0\,[\xi_\mathcal X,\zeta_\mathcal X]^\mu\bigl|_{\tau(\vec x)}\mathcal H_\mu(\vec x)\nonumber\\
&=-H([\xi_\mathcal X,\zeta_\mathcal X])\nonumber\\
&=H([\xi,\zeta])\;,\label{diffcovmommapequiv}
\end{align}
where the Lie bracket $[\xi,\zeta]$ is defined as the opposite of the commutator between the corresponding vector fields, say
\be
[\xi,\zeta]=-[\xi_\mathcal X,\zeta_\mathcal X]=-\left(\xi_\mathcal X^\nu{\zeta_\mathcal X^\mu}_{,\nu}-\zeta_\mathcal X^\nu{\xi_\mathcal X^\mu}_{,\nu}\right)\frac{\partial}{\partial x^\mu}\,.
\ee
Thus, the mapping $\xi\mapsto H(\xi)$ is a (anti)homomorphism between the Lie algebra $\textsf{diff}(\mathcal X)$ and the Poisson bracket algebra of observable functionals on the parametrized phase space. This shows that for any $\xi\in\textsf{diff}(\mathcal X)$ the functional $H(\xi)$ (resp. $\tilde{\mathcal J}_H(\xi)$) defines a equivariant momentum map w.r.t. the action of the spacetime diffeomorphism group on the parametrized phase space. Such property reflects the equivariance of the parametrized energy-momentum map \eqref{parametrizedemm}, which in turn was provided by considering the action \eqref{GactionEmb} of diffeomorphisms on the embeddings induced via the covariance fields. 

Finally, the total Hamiltonian \eqref{equivnewham} is constructed in such a way that the constraints are preserved along the flow generated by $H(\xi)$, that is
\be\label{constraintflow}
\dot{\mathcal H}_\alpha(\vec x)=\int_\Sigma\dd^nx_0^{\prime}\{\mathcal H_\mu(\vec x),\xi_\mathcal X^\nu(\tau(\vec x^{\,\prime}))\}\mathcal H_\nu(\vec x^{\,\prime})
\ee
vanishes on the constraint surface. As any functional of the embedding commutes with $H^{(\varphi)}$, we have
\be\label{embvar}
\dot{\tau}^\mu(\vec x)=\int_\Sigma\dd^nx'_0\,\xi_\mathcal X^\nu(\tau(\vec x^{\,\prime}))\{\tau^\mu(\vec x),P_\nu(\vec x^{\,\prime})\}=\xi_\mathcal X^\mu(\tau(\vec x))\;,
\ee
i.e., $\xi_\mathcal X(\tau(\vec x))$ is the deformation vector of the foliation which can be decomposed into its transverse $\vec \xi_\parallel(\tau(\vec x))$ and normal $\xi_\perp(\tau(\vec x))$ components. Unlike the Lagrange multipliers $N$ and $\vec N$ entering the parametrized action, the transverse and normal components of the deformation vector are now specific functionals of the embedding. Moreover, since $P(\xi)$ commutes with the field variables, the rates of change of the field $\varphi$ and its conjugate momentum $\Pi$ yield the Hamiltonian field equations with deformation vector $\xi_\mathcal X(\tau(\vec x))$.

Therefore, the co-momentum map $H(\xi)$ defined in Eq. \eqref{equivnewham} generates the deformation of the embedding induced by the vector field $\xi_\mathcal X$ on $\mathcal X$ together with the dynamical evolution of the field variables. Since the constraints are preserved along this flow, on-shell field configurations are compatibly evolved along the constraint surface. In other words, the canonical action of $\textsf{Diff}(\mathcal X)$ represented by $H(\xi)$ generates a displacement of the spatial hypersurface embedded in spacetime and it sets the consistently evolved Cauchy data for fields on the deformed hypersurface. 

The explicit embedding-dependence of the induced vector field $\xi_\mathcal X$ on $\mathcal X$, in  \eqref{equivnewham}, implies that $H(\xi)$ comes to be the Hamiltonian function along the flow lines generated by $\xi_\mathcal X$, which correspond to a one-parameter family of emebeddings identifying a foliation with deformation vector $\xi_\mathcal X(\tau(\vec x))$. Thus, for any $\xi\in\textsf{diff}(\mathcal X)$, realized as a complete vector field on $\mathcal X$, the induced vector field on $\textsf{Emb}_G(\Sigma,\mathcal X)$ is the tangent vector field to a curve of embeddings which identifies a foliation of spacetime. Different Lie algebra elements would identify different foliations, with corresponding deformation vectors given by the induced vector fields on $\mathcal X$ restricted to the embedding. In this sense the momentum map $H(\xi)$ (resp. $\tilde{\mathcal J}_H(\xi)$) provides a faithful representation of the Lie algebra of $\textsf{Diff}(\mathcal X)$ on the parametrized phase space\footnote{Consistently, as discussed in Appendix \ref{spatialdiffequilibrium}, by restricting to the diffeomorphisms which preserve the spatial slice $\Sigma$, the momentum map \eqref{equivnewham} reduces to \eqref{spacediffmommap}, which in turn yields a representation of the subalgebra of spatial diffeomorphisms.}.

\subsection{Generally Covariant Gibbs State} \label{covariantgibbs}

We now proceed to define a Lie-group generalization of the Gibbs equilibrium state for parametrized field theories, in terms of the Hamiltonian action of the spacetime diffeomorphism group $\textsf{Diff}(\mathcal X)$.

A statistical state $\rho:\Upsilon\to\mathbb R([0,+\infty[)$ on the parametrized phase space of the theory $\Upsilon\equiv T^*\mathscr Y_\tau\times T^*\textsf{Emb}_G(\Sigma,\mathcal X)$, is a smooth probability density on $\Upsilon$ such that, for any Borel subset $\mathscr A$ of $\Upsilon$, the integral
\be
\mu(\mathscr A)=\int_{\mathscr A}\,\mathcal D[\varphi,\Pi,\tau,P]\,\rho(\varphi,\Pi,\tau,P)
\ee
defines a probability measure on $\Upsilon$ with the normalization condition
\be\label{CPFTnormaliz}
Z(\rho)=\int_\Upsilon\,\mathcal D[\varphi,\Pi,\tau,P]\,\rho(\varphi,\Pi,\tau,P)=1\;,
\ee
where $\mathcal D[\varphi,\Pi,\tau,P]$ formally denotes the integration measure on $\Upsilon$ (assumed to be $\textsf{Diff}(\mathcal X)$-invariant). To such a statistical state, we can associate a entropy functional
\be\label{CPFTentropy}
S(\rho)=-\int_\Upsilon\mathcal D[\varphi,\Pi,\tau,P]\,\rho(\varphi,\Pi,\tau,P)\log\rho(\varphi,\Pi,\tau,P)\;,
\ee
with the convention that $\rho\log\rho=0$ for $\rho=0$. Given a functional on $\Upsilon$, say $f\in\mathcal F(\Upsilon)$, the expected value of $f$ w.r.t. $\rho$ is then defined as
\be
\mathbb{E}_\rho[f]=\int_\Upsilon\mathcal D[\varphi,\Pi,\tau,P]\rho(\varphi,\Pi,\tau,P)f(\varphi,\Pi,\tau,P):\mathcal F(\Upsilon)\longrightarrow\mathbb R \ .
\ee

We derive the Gibbs state via \emph{maximum entropy principle} \cite{Jaynes}, 
by asking for the stationarity of the entropy functional \eqref{CPFTentropy} under infinitesimal smooth variation $\rho_s(\varphi,\Pi,\tau,P)$ with $s\in]-\varepsilon,\varepsilon[$, $\varepsilon>0$ of the statistical state $\rho$, for fixed expected value of the co-momentum map, namely 
\be
\mathbb E_\rho(\tilde{\mathcal J}_H(\xi))=\int_\Upsilon\mathcal D[\varphi,\Pi,\tau,P]\rho(\varphi,\Pi,\tau,P)\braket{\xi(\tau),\tilde{\mathcal J}_H(\varphi,\Pi,\tau,P)}= const \ .
\ee
This is implemented by introducing two \emph{real} Lagrange multipliers 
$a,b \in\mathbb R$ respectively associated to the normalization condition \eqref{CPFTnormaliz} and the constraint $\mathbb{E}_\rho(\tilde{\mathcal J}_H(\xi))=const$, via
\be\label{con2}
\mathcal S(\rho_s)=S(\rho_s)+b\,\mathbb E_{\rho_s}(\tilde{\mathcal J}_H(\xi))+a\,Z(\rho_s)\;,
\ee
such that
\be
\frac{\delta\mathcal S(\rho_s)}{\delta s}\biggl|_{s=0}=0\qquad\forall\;\rho_s\;.
\ee
Hence, we get
\begin{align}\label{kkk}
0&=\frac{\delta\mathcal S(\rho_s)}{\delta s}\biggl|_{s=0}\nonumber\\
&=-\int_\Upsilon\mathcal D[\varphi,\Pi,\tau,P]\Bigl(1+\log\left(\rho(\varphi,\Pi,\tau,P)\right)- b\,\tilde{\mathcal J}_H(\xi)(\varphi,\Pi,\tau,P)-a\Bigr)\frac{\delta\rho_s}{\delta s}\biggl|_{s=0}\;,
\end{align}
for any $\rho_s$, from which it follows that
\be
\rho_{a,b}(\varphi,\Pi,\tau,P)=\exp{\Bigl(-1+a+b\braket{\xi(\tau),\tilde{\mathcal J}_H(\varphi,\Pi,\tau,P)}\Bigr)}.
\ee
The normalization condition \eqref{CPFTnormaliz} then implies
\be\label{CPFTpartitionfunction}
Z(\xi,b)=\exp{(1-a)}=\int_\Upsilon\mathcal D[\varphi,\Pi,\tau,P]\exp{\Bigl(b\,\tilde{\mathcal J}_H(\xi)(\varphi,\Pi,\tau,P)}\Bigr)\,,
\ee
where we limit $b$ to the set of values such that the above integral \emph{converge}. The generally covariant Gibbs statistical state is then given by
\begin{align}
\rho^{\text{(eq)}}_{\xi_{(b)}}(\varphi,\Pi,\tau,P)&=\frac{1}{Z(\xi,b)}\exp{\Bigl(b\,\braket{\xi,\tilde{\mathcal J}_H(\varphi,\Pi,\tau,P)}\Bigr)}\nonumber\\
&=\frac{1}{Z(\xi,b)}\exp{\left(-\int_{\Sigma}\dd^nx_0\,\xi_{(b)}^\mu(\tau(\vec x))\left(\mathcal H^{(\varphi)}_\mu(\vec x)+P_\mu(\vec x)\right)\right)}
\label{covarianteqstate}
\end{align}
The equilibrium temperature $b$ acts as a \textit{global} scaling of $\xi \in\textsf{diff}(\mathcal X)$, which can be interpreted as a geometric generalization (algebra-valued) of a \emph{local} equilibrium temperature and depends of the realization of the algebra on the embedding space (cfr. \eqref{Embvectorfield}). The product $\xi_{(b)}=b\,\xi \in\textsf{diff}(\mathcal X)$ defines a globally scaled vector field on $\mathcal X$ generating a one-parameter family of spacetime diffeomorphisms (scaled thermal time).

\begin{remark}
The result in \eqref{covarianteqstate} is close in form to Souriau's generalized Gibbs state (cfr. Sec. \ref{LGT}), but different in a key aspect due to covariance. In fact, in Lie Group Thermodynamics, the variational principle requires the momentum map associated to the $\mathcal{G}$-slicing preserving - one-parameter subgroup of the diffeomorphisms to be constant. Because of
$$
\mathbb E_\rho(\tilde{\mathcal J}_H(\xi))= \braket{\xi,\mathbb E_\rho(\tilde{\mathcal J}_H)} \ ,
$$ 
one can use the element of the algebra $\xi \in \textsf{diff}(\mathcal X)$ as the Lagrangian multiplier which eventually is identified with a global (inverse) temperature vector. Differently, in our general covariant approach, the explicit dependence on the embedding space in the realization of the algebra vectors implies that
\be
\mathbb E_\rho(\tilde{\mathcal J}_H(\xi))\neq \braket{\xi,\mathbb E_\rho(\tilde{\mathcal J}_H)} \ ,
\ee
hence the Lie group generalization of the stationarity condition in \eqref{kkk} does not work. In particular, the dependence of the algebra realizations on the embeddings space makes Souriau's geometric temperature local. Nevertheless, whenever the generalized Gibbs states \eqref{covarianteqstate} can be defined, we can think of it as the result of the exponentiation of a \emph{local} thermodynamic eqilibrium state w.r.t. a constant $b$, 
\begin{align}
\left(\exp{\Bigl(\braket{\xi,\tilde{\mathcal J}_H(\varphi,\Pi,\tau,P)}\Bigr)}\right)^b
\label{covarianteqstate}
\end{align}
multiplied by a $b$-dependent factor that preserves the normalization: $\rho_b = Z^{-1}(b)\, \rho^b$ \cite{43}. 
The effect of this exponentiation is to scale the thermal time globally, and therefore to scale the temperature globally.
\end{remark} 

The statistical state \eqref{covarianteqstate} is now a functional of the fields $(\varphi,\Pi,\tau,P)$ through the scaled comomentum map functional $b\,\braket{\xi,\tilde{\mathcal J}_H}$. In particular, being a functional on the parametrized phase space, the dependence from the spacetime coordinates occurs only through the dynamical variables thus respecting the coordinate-independence of relativistic theories. Moreover, as a functional of the embeddings, the statistical state \eqref{covarianteqstate} is covariant in the sense that the momentum map is evaluated on any space-like hyper-surface without fixing the slicing a priori. The one-parameter group of automorphisms of the extended configuration space generated by $\xi_{(b)}\in\textsf{diff}(\mathcal X)$ identifies a generalized concept of ``time evolution'' w.r.t. which the Gibbs state is of equilibrium\footnote{As discussed in Appendix \ref{spatialdiffequilibrium}, the restriction to the subalgebra of hyper-surface preserving spatial diffeomorphisms yields an equilibrium Gibbs state under one-parameter subgroups of spatial diffeomorphisms.}.

\subsection{Covariant Gauge Group Thermodynamics}\label{GTF}

Let us assume that the Hamiltonian action of the -- $\mathcal{G}$-slicing preserving -- one-parameter subgroup of diffeomorphisms, and its covariant moment map $\tilde{\mathcal J}_H:\tilde{\mathscr{P}}_\tau\longmapsto\mathfrak g^*$ are such that the ensemble of generalized (inverse) temperatures $\Omega \subset \mathbb{R}$ is non-empty.

Within the global equilibrium setting, the standard thermodynamic relations can be derived in terms of partial derivatives in $b$ of the partition function $Z(b)$, with the co-momentum map playing the role of the energy of the fields in thermodynamics . 

Given $Z(\xi,b)$, the free energy potential, $F(\xi,b)\equiv -\log Z(\xi,b)$, encodes complete thermodynamic information about the system. The \emph{equilibrium} internal energy $Q:\Omega\to\mathbb{R}$ is given by the gradient of $F(b)$
\begin{align}
Q(\xi,b)&=-\partial_b(\log Z(\xi,b))\nonumber\\
&=-\int_\Upsilon\mathcal D[\varphi,\Pi,\tau,P]\rho_{\xi_{(b)}}(\varphi,\Pi,\tau,P)\,\tilde{\mathcal J}_H(\xi)(\varphi,\Pi,\tau,P)\nonumber\\
&=-\mathbb{E}_{\rho_{\xi_{(b)}}}[\tilde{\mathcal J}_H(\xi)]
=\mathbb{E}_{\rho_{\xi_{(b)}}}[H^{(\varphi)}(\xi)]+\mathbb{E}_{\rho_{\xi_{(b)}}}[P(\xi)]\;,\label{inten}
\end{align}
corresponding to the average co-momentum map in the generalised Gibbs ensemble. Note that the internal energy is composed of two contributions, respectively associated to matter fields and embeddings. However, the two are not independent, as a dependence on the embeddings occurs on both terms in the r.h.s. of \eqref{inten} via the Gibbs state averaging. 
Exactly as for standard equilibrium thermodynamics, to a given $Q$ corresponds at most one value of $\xi_b$, so that $F(\xi,b)$ and the equilibrium probability density $\rho_{\xi_{(b)}}$ are uniquely determined. 

The entropy function $S: \Omega \to \mathbb{R}$, as defined in \eqref{CPFTentropy}, has a strict maximum at equilibrium given by
\be\label{eqent}
 S(\xi,b)=-F(\xi,b) - b\, \mathbb{E}_{\rho_{\xi_{(b)}}}[\tilde{\mathcal{J}}_H(\xi)] = -F(\xi,b) + b\, \mathbb{E}_{\rho_{\xi_{(b)}}}[{H}^{(\varphi)}(\xi)]+ b\, \mathbb{E}_{\rho_{\xi_{(b)}}}[P(\xi)]\;.
 \ee
Note that, at equilibrium, the equivariance of the thermodynamic potentials is preserved. 



Thermodynamic relations measure changes in such potentials, induced by infinitesimal departures from equilibrium 
(see e.g. \cite{jarz}). 
 We evolve the system from one equilibrium state to another, while generally driving it away from equilibrium. In the global equilibrium setting, this can be achieved either via a \emph{isothermal} transformation, which leaves the global temperature unchanged while moving the system on a different gauge orbit ($\xi_i \neq \xi_f$), or via a \emph{adiabatic} transformation, corresponding to an overall shift along the same gauge orbit. 
 
 A convenient measure of the difference between initial and final equilibrium states is given by Kullback-Leibler divergence \cite{amari}, which in our setting reads
\begin{equation*}
  D(\rho_{i}|\rho_{f})=\int_\Upsilon\,\mathcal D[\sigma]  \rho_{i}(\sigma) \log\left(\frac{\rho_i(\sigma)}{\rho_f(\sigma)}\right) \ge 0\ ,
\end{equation*}
where $\sigma$ synthetically denotes the elements in $\Upsilon\equiv T^*\mathscr Y_\tau\times T^*\textsf{Emb}_G(\Sigma,\mathcal X)$ as to ease the notation.

\noindent
For the isothermal case, we have 
\begin{eqnarray}\label{line}
D(\rho_{i,b}|\rho_{f,b})&=&\int_\Upsilon\,\mathcal D[\sigma] \rho_{i,b}(\sigma)\log\left(\frac{e^{ b\,\tilde{\mathcal J}_{H}(\xi_i)}/Z_i(b)}{e^{ b\,\tilde{\mathcal J}_{H}(\xi_f)}/Z_f(b)}\right)\\ 
&=& -\log Z_i(b) + b\,\mathbb{E}_{\rho_{i,b}}[\tilde{\mathcal J}_{H}(\xi_i)] +\log Z_f(b) - b\, \mathbb{E}_{\rho_{i,b}} [\tilde{\mathcal J}_{H}(\xi_f) ] \\ \nonumber
&=&\log Z_f(b) - \log Z_i(b) - b\,  \mathbb{E}_{\rho_{i,b}}[\tilde{\mathcal J}_{H}(\xi_f) -\tilde{\mathcal J}_{H}(\xi_i)] \ge 0
\end{eqnarray}
so that, by interpreting $W\equiv - \left(\tilde{\mathcal J}_{H}(\xi_f) -\tilde{\mathcal J}_{H}(\xi_i)\right)$ as the work  associated to the change of the element of the Lie algebra $\xi_i \rightarrow \xi_f $, we have
\bes
b\, \mathbb{E}_{\rho_{i,b}}(W) \ge \Delta F \equiv F_{f}(b) -F_{i}(b) \ .
\ees
The positivity of the divergence tells us that the external work performed on the system is no less than the free energy difference between the initial and final state. 

 Notice that, at the infinitesimal level, the work is noting but the co-momentum map for the Lie-bracket, namely $W=-\tilde{\mathcal J}_{H}([\xi_f, \xi_i])$ . Given the equivariance \eqref{diffcovmommapequiv} of the co-momentum map, this is the Poisson-bracket between the two co-momentum maps associated to the two Lie algebra elements $\xi_i$ and $\xi_f$, namely $ W =-\{\tilde{\mathcal J}_{H}(\xi_i), \tilde{\mathcal J}_{H}(\xi_f)\}$. 
 When the parameter is varied slowly enough so that the system remains in equilibrium along the flow, then the process is reversible and isothermal, and $\mathbb{E}_{i,b}( W) = \Delta F$. The case where the same two states at $[i,b]$ and $[f,b]$ are related via the action of a $\mathcal{G}$-slicing preserving one-d diffeomorphism subgroup is apparent: moving along an coadjoint orbit in $\mathfrak{g}^*$, by definition, we have $\Delta F=0$.
\

In case of \emph{adiabatic} transformation, where we imagine to keep $\tilde{\mathcal J}$ (or $\xi$) fixed, while changing the temperature $b_i \rightarrow b_f$, we can measure the distance between initial and final states starting from \eqref{line}, and applying the general definition of the entropy of the Gibbs state, $S(b)= \log Z(b)-b\,  \mathbb{E}_{\rho_b}[\tilde{\mathcal J}_H(\xi)]$ to get
\begin{eqnarray}
S(b_f)-S(b_i) &\ge&  - b_f \, \mathbb{E}_{\rho_{b_f}}[\tilde{\mathcal J}_H(\xi)]+ b_f\, \mathbb{E}_{\rho_{b_i}}[\tilde{\mathcal J}_H(\xi)]\\ \nonumber
&=& - b_f\, \left(  \mathbb{E}_{b_f}[\tilde{\mathcal J}_H(\xi)]- \mathbb{E}_{b_i}[\tilde{\mathcal J}_H(\xi)]\right) \\ \nonumber
&=&  - b_f\, \Delta Q\;.
\end{eqnarray}
Hence, by changing the global temperature, while keeping the reference observer/foliation associated with the generating vector field $\xi$, the entropy of the system increases, while the reservoir entropy decreases by an amount $b_f \Delta Q$, with an overall change in entropy    
\bes
\Delta S+ b_f \Delta Q \ge 0 \ .
\ees
This is the Clausius inequality of classical thermodynamics, which expresses the essential statement of the second law of thermodynamics.

A slow variation of the state realised via a series of equilibrium states defines a reversible \emph{adiabatic} transformation, with 
\be\label{dsss}\Delta S=- b_f \Delta Q = b_f \left(\Delta \langle H^{(\varphi)}(\xi)\rangle +\Delta \langle P(\xi)\rangle\right) \ .
\ee

The change in the global scaling $b_i\,\xi \to b_f\,\xi \in\textsf{diff}(\mathcal X)$ results in an increase of the entropy of the state of the fields. 
Differently form the isothermal process, here we do not perform a change on the embedding. In particular, the change in the entropy is not zero despite being measured along the same coadjoint orbit in $\mathfrak{g}^*$.\footnote{A remarkable difference w.r.t. the same relation when derived in the formal setting of Lie group thermodynamics (cfr. \cite{Souriau_b,Souriau,Marle}).}

\section{From Covariant Equilibrium to Dynamical Evolution}\label{EquiandDynam}


We are now interested in the relation between covariant equilibrium and dynamical evolution for the fields. In this sense, we are going to study how our generalized notion of Gibbs state on the extended phase space reduces by a gauge-fixing of the diffeomorphism symmetry, so as to disentangle the dynamics encoded in the constraints.

\subsection{Time Evolution Gibbs State via Gauge Fixing}\label{gaugefixedeqstate}

As discussed in Sec. \ref{diffreprmmap}, the equivariant co-momentum map in  \eqref{equivnewham} generates the evolution of the Cauchy data for the fields and the constraints among neighbouring spatial hypersurfaces. The vector field $\xi \in\textsf{diff}(\mathcal X)$ is the generating vector field of a spacetime foliation (cfr. Eq. \eqref{embvar}). Correspondingly, the one-parameter subgroup of diffeomorphisms associted with the scaled Lie algebra element $b\, \xi$ identifies a generalized notion of ``dynamical evolution'' w.r.t. which the covariant Gibbs state \eqref{covarianteqstate} is stable. 

As expected, for generally covariant systems, dynamics and gauge symmetry are deeply intertwined, in such a way that no preferred notion of time is available and the Hamiltonian is a combination of constraints. This makes the connection between the off- and on-shell levels of our covariant statistical analysis quite subtle. 


At first sight, it might seem difficult to extract the dynamical evolution of the physical fields as the Hamiltonian vanishes after imposing the constraints. A way to solve this issue is to introduce suitable gauge-fixing conditions such that the gauge evolution is frozen and the dynamics of fields can be disentangled.

As discussed in \cite{gaugefix1,gaugefix2}, in reparametrization invariant systems such as relativistic particles, general relativity or any diffeomorphism covariant field theory, at least one of the gauge-fixing conditions must depend on the time variable.\footnote{Here by time we mean the evolution parameter entering the canonical formulation of the theory.} This makes the application of the Dirac algorithm of constraints and gauge-fixing more delicate.\footnote{A detailed study of Hamiltonian formalism for systems with explicitly time-dependent second-class constraints has been carried out in \cite{timedepGF1,timedepGF2,timedepGF3} and reference within.} 

The key point is that in the case of time-dependent gauge-fixing constraints, the correct dynamical evolution for physical fields on the reduced phase space is generated by a new Hamiltonian which is not given by the restriction of the original Hamiltonian on the extended phase space (which vanishes on the constraint surface). The new Hamiltonian is determined by the \emph{time-dependent} Hamilton's equations of motions written in terms of the brackets on the reduced phase space induced by the Dirac bracket. The explicit form of this Hamiltonian of course depends on the details of the gauge-fixing conditions.

Following \cite{timedepGF1,timedepGF2,timedepGF3}, let us specialize the gauge-fixing procedure to the parametrized phase space of our theory $\Upsilon=T^*\mathscr Y_\tau\times T^*\textsf{Emb}_G(\Sigma,\mathcal X)$. Here, the canonical variables are $(\varphi^A(\vec x),\Pi_A(\vec x), \tau^\mu(\vec x),P_\mu(\vec x))$, $A=1,\dots,N$, $\mu=0,\dots,n$, with the first-class constraints $\mathcal H_\mu=P_\mu+\mathcal H_\mu^{(\varphi)}\approx0$ generating the gauge symmetries, and a Hamiltonian given by the co-momentum map $H(\xi)$ written in \eqref{equivnewham}. The canonical variables on $\Upsilon$ can be split into two disjoint sets $(\varphi^A,\Pi_A)$ and $(\tau^\mu,P_\mu)$, namely the physical fields of the theory (and their conjugate momenta) and the auxiliary field variables resulting from the parametrization procedure at the canonical level. The latter are the canonical variables we would like to eliminate by gauge-fixing. 

To this aim, let us introduce gauge-fixing conditions of the form
\be\label{gfconditions}
\chi^\mu=\tau^\mu-F^\mu(\varphi^A,\lambda)\quad,\qquad\mu=0,\dots,n
\ee
\noindent
where the $F^\mu$ are functions depending on the configuration field variables $\varphi^A$ as well as on the evolution parameter $\lambda$. The conditions $\chi^\mu=0$ provide us with a \emph{complete} set of gauge-fixing constraints. The set of all constraints, which will be denoted as $\psi_I=(\mathcal H_\mu,\chi^\nu)$, is now a second-class set (i.e., $\{\mathcal H_\mu,\chi^\nu\}\not\approx0$), and the conditions $\chi^\mu=0$ eliminate all the gauge freedom, namely
\be
\{\chi^\mu,\int\dd^nx_0\,\epsilon^\nu\mathcal H_\nu\}=0\quad,\quad\forall\lambda\quad\Rightarrow\quad\epsilon^\mu=0\;.
\ee
\noindent
In particular, the set of constraints $\psi_I=(\mathcal H_\mu,\chi^\nu)=0$ identifies the reduced phase space $\overline\Upsilon\subset\Upsilon$, where the gauge-fixed dynamics takes place. 
Denoting by $(\bar\varphi^A,\bar\Pi_A)$ a set of canonical variables on $\overline\Upsilon$, the restriction of the dynamical evolution to $\overline\Upsilon$ amounts to seek for a Hamiltonian $\bar H$ such that an equation of the kind
\be\label{reddynamics}
\frac{\dd f}{\dd\lambda}\approx\frac{\partial f}{\partial\lambda}+\{f,\bar H\}_*
\ee

\noindent
holds for any function $f(\bar\varphi^A,\bar\Pi_A,\lambda)$. Here, the weak equality $\approx$ denotes equality up to terms which vanish on $\overline\Upsilon$, and $\{\cdot,\cdot\}_*$ is the Dirac bracket defined by
\be
\{\cdot,\cdot\}_*:=\{\cdot,\cdot\}-\{\cdot,\psi_I\}\mathcal C^{IJ}\{\psi_J,\cdot\}
\ee
\noindent
with $\mathcal C^{IJ}=(\mathcal C)^{-1}_{IJ}$ the inverse of the matrix $\mathcal C$ whose entries are given by the Poisson brackets of all constraints, say $\mathcal C_{IJ}=\{\psi_I,\psi_J\}$\footnote{As $\psi_I=(\mathcal H_\mu,\chi^\nu)$ is a second class set, the matrix $\mathcal C$ is non-degenerate that is $\det\|\{\psi_I,\psi_J\}\|\neq0$.}.
By construction $\{f,\psi_I\}_*=0$ holds strongly for any second class constraint $\psi_I$ and any function $f$ on $\Upsilon$. All constraints (now second class) can thus be imposed strongly inside the bracket to eliminate the corresponding number of variables and obtain a well-defined induced bracket on $\overline\Upsilon$. The trajectories in $\Upsilon$ solving \eqref{reddynamics} project onto trajectories in $\overline\Upsilon$ lying in $\overline\Upsilon$ for all time as implied by the stability of the constraints $\frac{\dd\psi_I}{\dd\lambda}\approx0$. 

The Hamiltonian $\bar H$ is given by \cite{timedepGF1,timedepGF2,timedepGF3}
\be\label{newhamgf}
\bar H=H-\int_\Sigma\dd^nx_0\,\frac{\partial F^\mu}{\partial\lambda}P_\mu\;,
\ee

\noindent
which, as anticipated, is different from just $H$ restricted on $\overline\Upsilon$.
To see that $\bar H$ generates \eqref{reddynamics}, let us notice that since the gauge-fixing conditions \eqref{gfconditions} commute (strongly) among themselves, there exists a canonical transformation $(\varphi^A,\Pi_A,\tau^\mu,P_\mu)\mapsto(\bar\varphi^A,\bar\Pi_A,\bar Q^\mu,\bar P_\mu)$ such that $\bar Q^\mu=\chi^\mu$. The generating function of such a transformation is given by
\be
\mathfrak F(\varphi^A,\tau^\mu,\bar\Pi_A,\bar P_\mu,\lambda)=\varphi^A\bar\Pi_A+\Bigl(\tau^\mu-F^\mu(\varphi^A,\lambda)\Bigr)\bar P_\mu\;,
\ee

\noindent
so that we have
\be\label{cantransfgf}
\begin{cases}
\bar Q^\mu=\frac{\partial\mathfrak F}{\partial\bar P_\mu}=\tau^\mu-F^\mu(\varphi^A,\lambda)=\chi^\mu\\
\bar\varphi^A=\frac{\partial\mathfrak F}{\partial\bar\Pi_A}=\varphi^A\\
P_\mu=\frac{\partial\mathfrak F}{\partial\tau^\mu}=\bar P_\mu\\
\Pi_A=\frac{\partial\mathfrak F}{\partial\varphi^A}=\bar\Pi_A-\frac{\partial F^\mu}{\partial\varphi^A}\bar P_\mu
\end{cases}
\ee

\noindent
and Hamiltonian given by
\be
\bar H=H+\int_\Sigma\dd^nx_0\,\frac{\partial\mathfrak F}{\partial\lambda}=H-\int_\Sigma\dd^nx_0\,\frac{\partial F^\mu}{\partial\lambda}\bar P_\mu\;.
\ee

\noindent
In the new variables, the gauge-fixing conditions are now part of the field configuration variables and have no explicit $\lambda$-dependence anymore. Hence, stability of the gauge-fixing constraints simply amounts to require that
\be\label{newstabilitygf}
\dot{\bar Q}^\mu=\frac{\dd\chi^\mu}{\dd\lambda}=\{\chi^\mu,\bar H\}\approx0\;.
\ee

\noindent
As can be checked by direct computation, from the constraint algebra it follows that the matrix $\mathcal C=\|\{\psi_I,\psi_J\}\|$ takes the block form
\be
\mathcal C_{IJ}=
\left(
\begin{array}{c|c}
0 & \;\;\mathcal A_{\mu\nu} \\
\hline
-\mathcal A_{\mu\nu} & \;\;0
\end{array}
\right)\qquad,\qquad\mathcal A_{\mu\nu}=\{\mathcal H_\mu,\chi_\nu\}
\ee
\noindent
so that for any function $f$ on $\Upsilon$ the Dirac bracket yields
\begin{align}\label{DBfham}
\{f,\mathcal{\bar{H}}\}_*&=\{f,\mathcal{\bar{H}}\}+\sum_{\mu,\nu}\Bigl(\{f,\mathcal H_\mu\}(\mathcal A^{-1})_{\mu\nu}\{\chi_\nu,\mathcal{\bar H}\}-\{f,\chi_\mu\}(\mathcal A^{-1})_{\mu\nu}\{\mathcal H_\nu,\mathcal{\bar H}\}\Bigr)\nonumber\\
&\approx\{f,\mathcal{\bar{H}}\}-\sum_{\mu,\nu}\{f,\chi_\mu\}(\mathcal A^{-1})_{\mu\nu}\{\mathcal H_\nu,\mathcal{\bar H}\}\;,
\end{align}
\noindent
where in the last line we used Eq. \eqref{newstabilitygf}. Now, as $\bar Q^\mu=\chi^\mu$, the second term on the r.h.s. of \eqref{DBfham} vanishes if we restrict $f=f(\bar\varphi^A,\bar\Pi_A,\lambda)$. Thus, for any function $f(\bar\varphi^A,\bar\Pi_A,\lambda)$ we get
\be
\frac{\dd f}{\dd\lambda}=\frac{\partial f}{\partial\lambda}+\{f,\bar H\}\approx\frac{\partial f}{\partial\lambda}+\{f,\bar H\}_*\;,
\ee
\noindent
which is the desired form \eqref{reddynamics} of the dynamical evolution.

At this stage, we can further specify $F^\mu$ in the gauge-fixing conditions \eqref{gfconditions} so as to get the standard Hamiltonian for the fields on the reduced phase space.

Let us then consider the case in which the functions $F^\mu$ do not depend on the fields $\varphi^A$ but only on $\lambda$, say $\chi^\mu=\tau^\mu-F^\mu(\lambda)$\footnote{A specific realization of this situation is provided for instance by the gauge choice $F^\mu=x^\mu$ in which the functions $F^\mu$ depend only on $\lambda$ through the spacetime coordinates $x^\mu$.}. According to the expressions \eqref{cantransfgf}, in this case the canonical variables on $\overline\Upsilon$ are just given by the physical fields and their conjugate momenta, i.e.:
\be
\bar\varphi^A=\varphi^A\qquad,\qquad\bar\Pi_A=\Pi_A\;.
\ee
\noindent
The constraints $\mathcal H_\mu=0$ can be then solved to express the momenta $P_\mu=\bar P_\mu$ in terms of the variables $(\bar\varphi^A,\bar\Pi_A)$ yielding $P_\mu=-\mathcal H^{(\varphi)}_\mu$. The new Hamiltonian \eqref{newhamgf} then reads
\be
\bar H=H-\int_\Sigma\dd^nx_0\,\dot F^\mu \bar P_\mu\;,
\ee
\noindent
from which, restricting on $\overline\Upsilon$ and taking into account that $H=0$ imposing the constraints, it follows that 
\be\label{newgfham}
\bar H=\int_\Sigma\dd^nx_0\,\dot F^\mu\mathcal H^{(\varphi)}_\mu\;.
\ee
\noindent
Note that $\dot F^0$ and $\dot{\vec F}$ respectively play the role of lapse and shift consistently with $\xi^\mu(\tau(\vec x))=\dot\tau^\mu$ being the deformation vector field of the foliation which after gauge-fixing yields $\dot\tau^\mu=\dot F^\mu$. 

In particular, by choosing the gauge $F^\mu(\lambda)=x^\mu$ with
\be\label{timegauge}
x^0=\lambda\qquad,\qquad x^k=0\quad k=1,\dots,n
\ee
\noindent
the new Hamiltonian \eqref{newgfham} is nothing but the usual field Hamiltonian
\be
\bar H=\int_\Sigma\dd^nx_0\,\mathcal H^{(\varphi)}_0=H^{(\varphi)}\;,
\ee
\noindent
which as such generates the correct dynamics for the field variables $(\varphi^A,\Pi_A)$.

Finally, via the maximum entropy principle, we define the Gibbs state on the reduced phase space $\overline\Upsilon$ by
\be
\rho_{b,F}^{(\text{eq})}(\varphi,\Pi)=\frac{1}{Z(b,F)}\,e^{-b\,\bar H(\varphi,\Pi)}=\frac{1}{Z(b,F)}\,\exp\left(-\int_\Sigma\dd^nx_0\,b\,\dot F^\mu\mathcal H^{(\varphi)}_\mu\right)\;,
\ee
\noindent
with
\be
Z(b,F)=\int_{\overline\Upsilon}\mathcal D[\varphi,\Pi]\,e^{-b\, \bar H(\varphi,\Pi)}\;.
\ee
\noindent
This is an equilibrium state w.r.t. the dynamical evolution generated by the Hamiltonian $\bar H$. In particular, choosing the temporal gauge $F^0=x^0$, $\vec F=\vec0$, we get the standard relativistic equilibrium state with local (as being a function of $\lambda$) temperature given by the inverse of the lapse. In the special case for $x_0=\lambda$, the lapse is one and the (inverse) temperature is globally identified with $b$.

%

\subsection{On the Thermodynamic Characterization of Covariant Equilibrium}\label{thermalinterpret}

Let us close this section by commenting on the relation between the generalized covariant Gibbs state for parametrized field theories presented in this work and the thermodynamic characterization of covariant statistical equilibrium arising in previous investigations. This not only provides us with a consistency check for our formalism, but also gives some insight on the physical interpretation of the Lie algebra-valued temperature in such a framework. In particular, we elaborate on its inbuilt observer dependence and its similarities as well as differences with the thermal time hypothesis.

The essence of thermal time \cite{40,41,42} relies on a fundamental reinterpretation of the relation between equilibrium states and time flow, according to which any statistical state $\rho$ is in equilibrium w.r.t. its own modular flow. The modular Hamiltonian $H=-\log\rho$ generates a one-parameter group of transformations, which defines the time flow associated to the state $\rho$. In particular, physical equilibrium states are those whose thermal time identifies a flow in spacetime \cite{43}, thus providing a thermodynamical characterization of the notion of time experienced by an observer \cite{40}. 

In our extended parametrised phase space description,  for any Lie algebra element $\xi$ the co-momentum map $H(\xi)$ defined in \eqref{equivnewham} can be thought of as the thermal modular Hamiltonian
\begin{equation}\label{covmodularham}
H(\xi)(\varphi, \Pi, \tau, P)=-\braket{\xi,\tilde{\mathcal J}_H(\varphi, \Pi, \tau, P)}=\log\rho_\xi(\varphi, \Pi, \tau, P)
\end{equation}
associated with the covariant state $\rho_{\xi}$. The corresponding thermal flow identifies a one-parameter group on $\Upsilon\equiv T^*\mathscr Y_\tau\times T^*\textsf{Emb}_G(\Sigma,\mathcal X)$ generated by the Hamiltonian vector field $X_H$ satisfying
\be
\rho_\xi\text{i}_{X_H}\tilde\omega=\dd\rho_\xi\qquad\text{or equivalently}\qquad \text{i}_{X_H}\tilde\omega=\dd\log\rho_\xi=\dd H(\xi)\;,
\ee
where $\tilde\omega$ is the symplectic structure on $\Upsilon$ given in \eqref{fieldsymplecticstructure} and, with a slight abuse of notation, we still denote by $\dd$ the exterior differential on the extended parametrized phase space of fields although it should not be confused with that on spacetime. More explicitly, as already anticipated in Sec. \ref{diffreprmmap}, the equations of motion associated to the Hamiltonian $H(\xi)$ are given by
\begin{align}
\dot\varphi^A(\vec x)&=\{\varphi^A(\vec x),H(\xi)\}=\int_\Sigma \dd^nx_0'\xi^\nu(\tau(\vec x\,'))\{\varphi^A(\vec x), \mathcal H_\nu^{(\varphi)}(\vec x\,')\}=\xi^\nu(\tau(\vec x))\frac{\delta\mathcal H_\nu^{(\varphi)}}{\delta\Pi_A}\;,\label{eom1}\\
\dot\Pi_A(\vec x)&=\{\Pi_A(\vec x),H(\xi)\}=\int_\Sigma \dd^nx_0'\xi^\nu(\tau(\vec x\,'))\{\Pi_A(\vec x), \mathcal H_\nu^{(\varphi)}(\vec x\,')\}=-\xi^\nu(\tau(\vec x))\frac{\delta\mathcal H_\nu^{(\varphi)}}{\delta\varphi^A}\;,\label{eom2}\\
\dot\tau^\mu(\vec x)&=\{\tau^\mu(\vec x),H(\xi)\}=\int_\Sigma \dd^nx_0'\xi^\nu(\tau(\vec x\,'))\{\tau^\mu(\vec x), P_\nu(\vec x\,')\}=\xi^\mu(\tau(\vec x))\;,\label{eom3}\\
\dot P_\mu(\vec x)&=\{P_\mu(\vec x),H(\xi)\}\approx\int_\Sigma \dd^nx_0'\xi^\nu(\tau(\vec x\,'))\{P_\mu(\vec x), P_\nu(\vec x\,')\}=0\label{eom4}\;,
\end{align}
so we see that the vector field
\begin{align}
X_H&=\dot\varphi^A\frac{\delta}{\delta\varphi^A}+\dot\tau^\mu\frac{\delta}{\delta\tau^\mu}+\dot\Pi_A\frac{\delta}{\delta\Pi_A}+\dot P_\mu\frac{\delta}{\delta P_\mu}\nonumber\\
&=\xi^\mu(\tau(\vec x))\left[\left(\frac{\delta\mathcal H_\mu^{(\varphi)}}{\delta\Pi_A}\right)\frac{\delta}{\delta\varphi^A}-\left(\frac{\delta\mathcal H_\mu^{(\varphi)}}{\delta\varphi^A}\right)\frac{\delta}{\delta\Pi_A}+\frac{\delta}{\delta\tau^\mu}\right]\;,
\end{align}
generates the correct dynamical evolution for the fields as well as for the constraints (see Eq. \eqref{constraintflow}). 

In particular, looking at $\Upsilon$ as a bundle over $\textsf{Emb}(\Sigma,\mathcal X)$, whose fiber over $\tau\in\textsf{Emb}(\Sigma,\mathcal X)$ is the field phase space $T^*\mathscr{Y}_\tau$, coordinatized by the spatial matter fields and their conjugate momenta, we see from \eqref{eom3} that the above vector field restricted to $\textsf{Emb}(\Sigma,\mathcal X)$ generates a one-parameter curve of embeddings $c:\mathbb R\to\textsf{Emb}(\Sigma,\mathcal X)$ by $c(\lambda)=\tau(\lambda)$, which in turn identifies a slicing $\mathfrak{s}_\mathcal X(\vec x,\lambda)=\tau(\lambda)(\vec x)$ in spacetime generated by the vector field $\xi_\mathcal X(\tau(\vec x))$. The corresponding curve $c_\tau(\lambda)=(\varphi(\lambda),\Pi(\lambda))$ in $T^*\mathscr{Y}_\tau$ identifies then a compatible $\mathcal G$-slicing of the field space. 

As schematically pictured in Fig.\ref{thflow}, the thermal flow associated to the state $\rho_{\xi}$ defined in \eqref{covmodularham} with thermal time parameter $\lambda$ defines a one-parameter group of bundle automorphisms as well as a compatible slicing of spacetime so that the Lie algebra-valued temperature $\xi$ identifies a direction in spacetime along which geometry and matter fields evolve. 

\begin{figure}[t!]
\begin{center}
\includegraphics[scale=0.42]{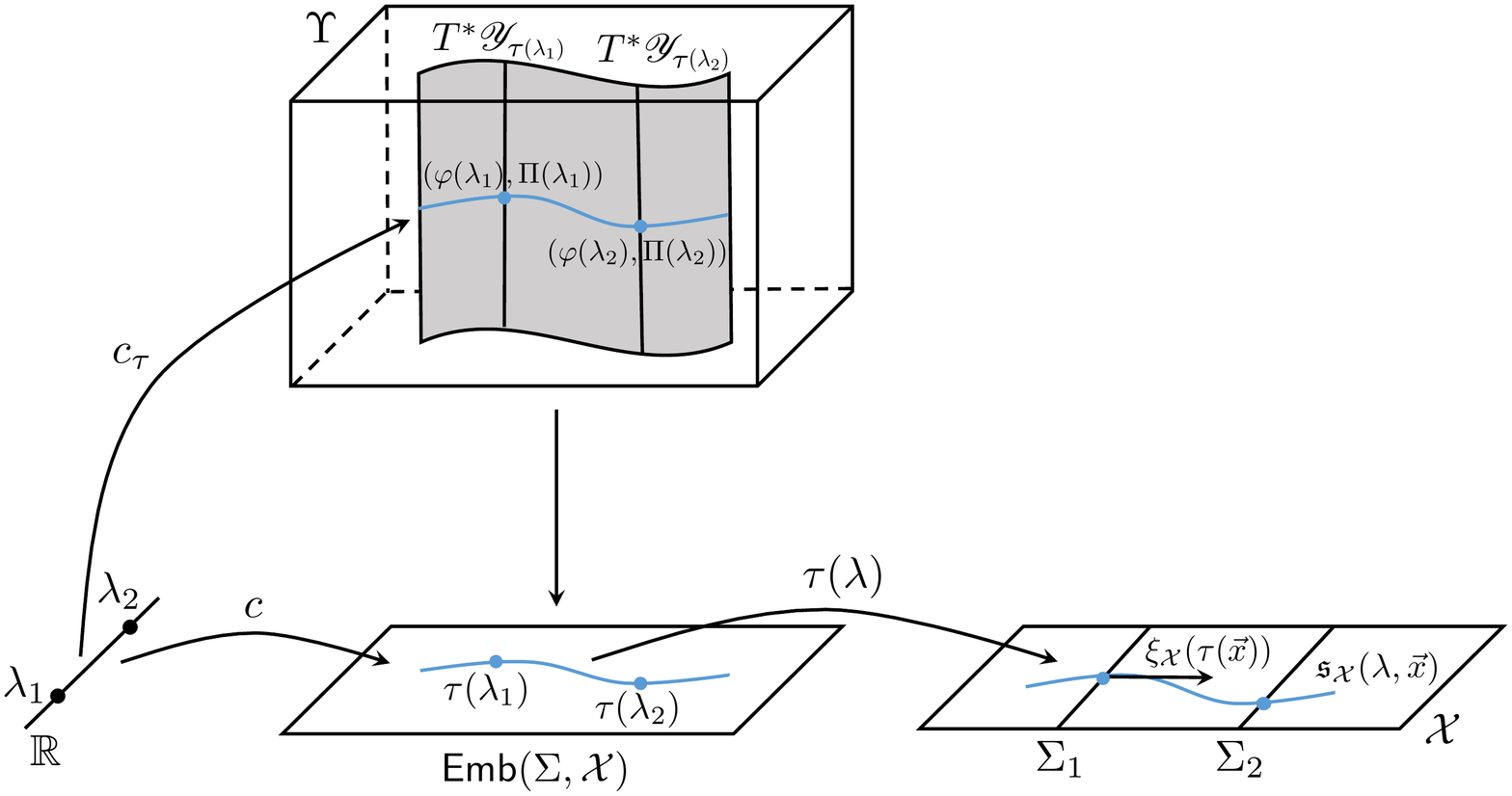}\caption{Slicing of spacetime and of the field bundles over it generated by the thermal flow associated to the Gibbs state \eqref{covarianteqstate}. The Lie algebra-valued temperature $\xi$ thus defines a generalized notion of direction of dynamical evolution characterized by a ``thermal time parameter'' $\lambda$.}
\label{thflow}
\end{center}
\end{figure}

Now, in our covariant statistical mechanic approach, the notion of covariant modular equilibrium associated to $\rho_{\xi}$ is related with the notion of thermodynamic equilibrium, via the Gibbs state derived in \eqref{covarianteqstate}, modulo the introduction of a global temperature rescaling $b$. 
In this sense, the Lie algebra element $\xi$ provides a \emph{local} notion of temperature.
Indeed, along the line of \cite{43}, we can define a local temperature function, i.e. a map $T: \Upsilon\times\Sigma\mapsto\mathbb R$, by
\be
T((\varphi,\Pi,\tau,P),\vec x)=T(\varphi(\vec x),\Pi(\vec x),\tau(\vec x),P(\vec x)):=|\xi(\tau(\vec x))|^{-1}\;,
\ee
with the vector field $\xi^{\mu}(\tau(\vec x))$ playing the role of \textit{multi-fingered time}. Consistently, with the gauge choice $F^0=x^0$ discussed in the previous section, the above temperature function reduces on-shell to the inverse of the lapse function, say $T=\frac{1}{\dot{x}^0}$, thus locally yielding the rate of change of the thermal time parameter $\lambda$ and the coordinate time $x^0$.

In our setting, the inbuilt observer dependence of the thermal time notion of equilibrium is made explicit.
In fact, the Lie algebra element $\xi$ identifies a foliation, and the associated one-parameter flow in spacetime characterizes an evolution flow associated with the corresponding canonical observer. Changing elements in the Lie algebra then corresponds to a change of observer frame, hence a different foliation identified by the corresponding one-parameter subgroup of diffeomorphisms. 

The main difference in comparing the above analysis with the framework of thermal time hypothesis is that the generalized Gibbs state \eqref{covarianteqstate} is defined on the parametrized canonical phase space and not on the space of solutions of the field equations. Only after introducing suitable gauge choices the thermal flow associated with the gauge-fixed Hamiltonian generates the on-shell dynamics for the matter fields (see Sec. \ref{gaugefixedeqstate}). In this sense, our framework can be thought of as an off-shell generalization of the thermal time setting, with the embedding dependence of the Lie algebra valued temperature playing a fundamental role in connecting the phase space and spacetime pictures.

%

\section{Conclusions and Outlook}\label{conclusions}

In this work, we consider the problem of defining statistical mechanics and thermodynamics for generally covariant field theories, where both dynamics and gauge symmetry are encoded in (first-class) constraints, hence no standard notions of time and energy are available. The key issue consists in the definition of statistical equilibrium, beyond time translations, for a symmetry group flow given by spacetime diffeomorphisms. Souriau's symplectic and covariant reformulation of statistical equilibrium for Hamiltonian Lie group actions provides a useful framework to this aim.\\
A multi-symplectic generalisation of Souriau's Lie group thermodynamics to the case of a fully constrained system suggests a definition of a Gibbs-like state with respect to the gauge group action generated by the first-class constraints. In this setting, the Hamiltonian Lie group action is characterized by the vanishing of the associated momentum map. The vanishing momentum map is equivalent to a first class constraint, reflecting the presence of a gauge symmetry for the system, and the constraint surface is identified with the zero-level set of the associated momentum map \cite{Henneaux, gimmsyIII, Giulini}. \\
However, in general-covariant theories the super-Hamiltonian and super-momentum constraints do not close a genuine Lie algebra which reflects into the non-equivariance of the energy-momentum map preventing a straightforward representation of the algebra of the spacetime diffeomorphism group. As a consequence, the application of Jaynes' entropy maximization principle \cite{Jaynes} with constant mean values of the first-class constraints leads to a Gibbs-like state which, due to the non-equivariance of the energy-momentum map, is not an equilibrium state under the gauge flow generated by the constraints.\footnote{In terms of the classification of different kinds of Gibbs states given in \cite{GFTgibbs}, this state would correspond to a \textit{thermodynamical} rather than \textit{dynamical} statistical state.} In particular, such a state is not of equilibrium with respect to a one-parameter group of spacetime diffeomorphisms, but only with respect to spatial diffeomorphisms, for which there is a momentum map induced by the energy-momentum map restricted to the Lie algebra of spatial diffeomorphisms (cfr. Eqs. \eqref{spacediffmommap} and \eqref{spatialmmap}).\\
We overcome this issue by \emph{extending} the multi-symplectic phase space description via the  introduction of the covariance fields. This allows us to recover the action of the spacetime diffeomorphism group by means of the observable Hamiltonian functional $H(\xi)$ defined in \eqref{equivnewham} which identifies an equivariant momentum map, homomorphically relating the diffeomorphism algebra to the Poisson bracket algebra of functionals on the parametrized phase space. The derivation of a Gibbs state associated to this equivariant momentum map via the prescriptions of Lie group thermodynamics eventually provides us with a \textit{dynamical} (in the sense of \cite{GFTgibbs}) equilibrium state on the parametrized phase space. Such a state is of equilibrium with respect to the one-parameter group of diffeomorphisms generated by the vector field $\xi_\mathcal X$ associated to $\xi\in\mathfrak g=\textsf{diff}(\mathcal X)$. In this sense, it defines a spacetime covariant notion of thermodynamical equilibrium.   

The gauge character of spacetime diffeomorphisms implies a radical conceptual shift in the definition of equilibrium state with respect to Souriau's work. By replacing a dynamical symmetry with a gauge one, we move our analysis from the fully reduced symplectic space of motions (on-shell) to the unconstrained extended phase space of the system (off-shell). While being defined off-shell, the covariant Gibbs state is by construction an observable of the theory, and it encodes, via the covariant momentum map functional, all the dynamical information carried by the given parametrized field theory: its canonical Hamiltonian, its initial value constraints, its gauge freedom, and its stress energy-momentum tensor \cite{gimmsyI}. Therefore, we expect the off-shell equilibrium to play a role similar to a generating functional partition function in field theory, with a Gibbs state corresponding to a ``soft'' imposition of the constraints of the theory. 

Indeed, for $\tilde{\mathcal J}_H(\varphi,\Pi,\tau,P) \to 0$, the Gibbs distribution converges to the uniform distribution defined over the reduced support with vanishing (though not necessarily minimal) momentum map, 
namely
\be
\rho^{\text{(eq)}}_b \rightarrow \delta(\bar{\sigma})
\ee
with $\bar{\sigma} \in \tilde{\mathcal J}_H^{-1}(0) \subset \tilde{\mathscr{P}}_\tau$, corresponding to field configurations satisfying both primary and secondary first class constraints. From the second Noether theorem, we know that field configurations $\bar{\sigma}$ such that ${\tilde{\mathcal J}_H}(\bar\sigma)=0$, represent \emph{solutions} to the Euler-Lagrange equations \cite{gimmsyI}. In this sense, we can think of reformulating the symplectic reduction problem for our first order parametrized field theory in thermodynamic terms, as a limiting case of a general principle of maximum entropy\footnote{For $b$ such that $\langle Q(b),b\rangle >0$, minimizing the free energy is the same as maximizing $S$.}  \cite{Jaynes}, and in accordance with the \emph{geometrodynamics regained} program of Kucha\v r \cite{regain}. To the idea that the geometry and symmetry together determine the field theory, the proposed derivation adds a further statistical characterisation, showing that from symmetry assumptions and the energy-momentum map it is possible to recover a sharper statistical information on the solutions of the Euler-Lagrange equations in some thermodynamic limit.

It is natural to understand the generalized Gibbs state in \eqref{covarianteqstate} as an \emph{off-shell} generalisation corresponding to a canonical statistical distribution in which a non-zero weight is assigned to configurations which do not solve the constraint equations. Interestingly, the passage from the canonical description to the microcanonical definition of the constrained theory space amounts to \emph{coarse-graining} the information encoded in the inverse  temperature vector in the algebra.\footnote{Indeed, one can see the uniform distribution partition function as a weighted superposition of canonical partition functions over various values of the intensive $b$ parameter.} Such a coarse-graining essentially erases the information on any specific choice of reference frame, corresponding to a sharper notion of space-time covariance. Choosing not to erase such information allows to see how an eventual deparametrisation or a gauge fixing, possibly physically induced, associates a notion of time to thermal equilibrium.


As our original motivation is general relativity, the present analysis needs to be extended to gravity for which a covariant Hamiltonian formalism has been extensively studied in the last years, see e.g. \cite{Rovelli2,50,c1,c1.2,c1.3,c1.4} and references within. However, the intrinsic parametrized nature of gravity \cite{c2} makes the achievement of a phase space representation of spacetime diffeomorphisms a non-trivial task. Useful insights in this respect may be found in the work by Isham and Kucha\v{r} \cite{ishamkuchar1}. A certainly interesting closely related aspect to explore consists in including boundaries in our covariant description (see e.g. \cite{c4a,c4b,c4c} for the analysis of boundaries in the framework of multi-symplectic field theories). The inclusion of boundaries would be needed to identify a notion of subsystems which is crucial in any thermodynamical analysis. As remarked above, in this covariant setting the constraints are beautifully encoded in the momentum map associated with the gauge group action on the field space. The latter plays a key role in the construction of the covariant Gibbs state. As shown in \cite{c5}, in presence of boundaries, the kinematical constraint algebra can be written as conservation laws for boundary charges. These charges would then enter the boundary momentum map out of which a statistical mechanic framework for the on-shell boundary modes can be constructed. This may open fruitful connections with recent work on boundary modes in quantum gravity \cite{c6}. Of particular interest would be the case in which the finite boundaries describe horizons, with potential application to black hole thermodynamics.\\
From a broader point of view, the application of our covariant statistical mechanics formalism to discrete gravity models may provide interesting insights to the study of coarse-graining approaches and the continuum limit from a thermodynamical perspective, as initiated for instance in \cite{c7,c8}. 


\section*{Acknowledgements}
The authors would like to thank Daniele Oriti and Isha Kotecha for discussions in the very early stages of the project. ML and FM thank AEI Potsdam for hospitality in the beginning of this work. The authors are also grateful to Florio M. Ciaglia and Fabio Di Cosmo for discussions and comments on an early version of the draft. FM is indebted with Johannes M\"unch for valuable discussions throughout this work. The work of FM at the University of Regensburg was supported by an International Junior Research Group grant of the Elite Network of Bavaria.

\appendix

\section{Equilibrium for Spacetime vs. Spatial Diffeomorphisms}\label{spatialdiffequilibrium}

As a consistency check for our formalism, it is interesting to investigate how the notion of the generalized equilibrium state changes when defined in terms of the one-parameter subgroups of full \emph{spacetime} and \emph{spatial} diffeomorphisms.

As discussed in Sec. \ref{diffreprmmap}, the action of the diffeomorphism group on the covariance fields induces a natural action of $\textsf{Diff}(\mathcal X)$ on the space of embeddings by left action \eqref{GactionEmb}. The natural action of the group $\textsf{Diff}(\Sigma)$ of \emph{spatial} diffeomorphisms on the space of embeddings is instead the right action, namely
\be\label{spatialdiffaction}
\textsf{Emb}(\Sigma,\mathcal X)\times\textsf{Diff}(\Sigma)\longrightarrow\textsf{Emb}(\Sigma,\mathcal X)\qquad\text{by}\qquad(\tau,\alpha_\Sigma)\longmapsto\tau\comp\alpha_\Sigma\;.
\ee
Accordingly, a point $\vec x\in\Sigma$ is mapped by $\alpha_\Sigma\in\textsf{Diff}(\Sigma)$ into a new point $\vec x'=\alpha_\Sigma(\vec x)\in\Sigma$, and then this point is mapped by the embedding $\tau$ into the spacetime point $\tau(\vec x')\in\mathcal X$.

The right action  of $\textsf{Diff}(\Sigma)$ \eqref{spatialdiffaction} induces a left action of $\textsf{Diff}(\mathcal X)$ on a given hypersurface $\tau(\Sigma)$ which preserves the hypersurface fixed, i.e., for a given $\tau\in\textsf{Emb}(\Sigma,\mathcal X)$, we have
\be\label{stdiffsdiff}
\tau\comp\alpha_\Sigma=\Phi_{\alpha_\Sigma}\comp\tau\;,
\ee
such that the following diagram
\be
\xymatrix{
\mathcal X\ar[r]^-{\Phi_{\alpha_{\Sigma}}} & \mathcal X\\
\Sigma\ar[u]^-{\tau}\ar[r]_-{\alpha_{\Sigma}} & \Sigma\ar[u]_-{\tau}
}
\ee
commutes. In particular, considering a one-parameter subgroup of $\textsf{Diff}(\Sigma)$ which induces a one-parameter subgroup of hypersurface preserving spacetime diffeomorphisms, the corresponding generating vector fields $\xi_\Sigma\in\mathfrak{X}(\Sigma)$ and $\xi_\mathcal X\in\mathfrak X(\mathcal X)$ are related by
\be\label{genvecfields}
\xi_\mathcal X=\tau_*\xi_\Sigma\qquad\text{i.e.}\qquad \xi_\mathcal X^\mu(\tau(\vec x))=\tau^\mu_{,k}(\vec x)\xi_\Sigma^k(\vec x)\;.
\ee
Differently from the fully spacetime covariant case, the embedding $\tau$ is now fixed, thus restricting the analysis to a given spatial hypersurface $\Sigma_\tau=\tau(\Sigma)$. Correspondingly, the restriction of the covariant momentum map $\tilde{\mathcal J}_H$ to the subalgebra $\mathfrak g_\tau=\textsf{diff}_{\Sigma_\tau}(\mathcal X)$ of spatial diffeomorphisms which preserves the image of $\tau$, identifies an equivariant momentum map w.r.t. the action of the subgroup $\mathcal G_\tau=\textsf{Diff}(\Sigma_\tau)$. Indeed, using the relation \eqref{genvecfields}, for $b_\tau\in\mathfrak g_\tau=\textsf{diff}_{\Sigma_\tau}(\mathcal X)$ we have
\begin{align}\label{spatialmmap}
\braket{\tilde{\mathcal J}_H(\varphi,\Pi,\tau,P),b_\tau}&=-\int_{\Sigma_\tau}\dd^nx_0\,\xi^\mu_{(b_\tau)}(\tau(\vec x))\mathcal H_\mu(\vec x)\nonumber\\
&=-\int_{\Sigma_\tau}\dd^nx_0\,\tau^\mu_{,k}(\vec x)\zeta^k_{(b_\tau)}(\vec x)\mathcal H_\mu(\vec x)\nonumber\\
&=-\int_{\Sigma_\tau}\dd^nx_0\,\zeta^k_{(b_\tau)}(\vec x)\mathcal H_k(\vec x)\nonumber\\
&=-(\vec P(\zeta)+\vec H^{(\varphi)}(\zeta))\nonumber\\
&=:\braket{\mathcal J_\tau(\varphi,\Pi,\tau,P),\zeta}\quad,\quad\zeta\in\textsf{diff}(\Sigma_\tau)
\end{align}

\noindent
so that, for $\zeta,\zeta'\in\textsf{diff}(\Sigma_\tau)$, we have

\be
\{\mathcal J_\tau(\zeta),\mathcal J_\tau(\zeta')\}=\{\vec H(\zeta),\vec H(\zeta')\}=\vec H([\zeta,\zeta'])\;,
\ee 

\noindent
as can be checked by direct computation using the fact that $\zeta^j(\vec x), \zeta'^k(\vec x')$ do not depend on the embeddings and the Poisson bracket for the spatial diffeomorphism constraint. Thus, the map $\mathcal J_\tau=\tilde{\mathcal J}_H\bigl|_{\Sigma_\tau}$ provides a representation of the algebra of spatial diffeomorphisms on the extended phase space. Consistently, the map \eqref{spatialmmap} gives the expected lift of the shift vector into the deformation vector tangential to the spatial slice, i.e.
\begin{align}
\dot\tau^\mu(\vec x)&=\{\tau^\mu(\vec x), \mathcal J_\tau(\zeta)\}=\{\tau^\mu(\vec x),\vec P(\zeta)\}\nonumber\\
&=\int_{\Sigma_\tau}\dd^nx'_0\,\zeta^k(\vec x')\tau^\nu_{,k}(\vec x')\{\tau^\mu(\vec x),P_\nu(\vec x')\}\nonumber\\
&=\tau^\mu_{,k}(\vec x)\zeta^k(\vec x)\;,
\end{align}
and the momentum map $\mathcal J_\tau=\tilde{\mathcal J}_H\bigl|_{\Sigma_\tau}$ restricted to the subalgebra of hyper-surface preserving diffeomorphisms coincides with the usual $\textsf{diff}(\Sigma_\tau)$-equivariant momentum map \eqref{spacediffmommap} with $\zeta^k(\vec x)$ playing the role of the shift vector. In other words, being the embedding fixed, we are now ``constraining'' the lifted action of the diffeomorphism group on the space of embeddings in such a way that the induced spacetime diffeomorphism preserves the spatial hypersurface thus yielding an equivariant momentum map under spatial diffeomorphisms only and not under the full spacetime diffeomorphisms group.

Accordingly, restricting ourselves to the subgroup of spatial diffeomorphisms, the covariant Gibbs state \eqref{covarianteqstate} reduces to
\begin{align} \label{spatialdiffeqstate}
\rho^{\text{(eq)}}_{b_\tau}(\varphi,\Pi,\tau,P)&=\frac{1}{Z(b_\tau)}\exp{\Bigl(-\braket{b_\tau,\tilde{\mathcal J}_H(\varphi,\Pi,\tau,P)}\Bigr)}\nonumber\\
&=\frac{1}{Z(\zeta)}\exp{\Bigl(-\braket{\zeta,\mathcal J_\tau(\varphi,\Pi,\tau,P)}\Bigr)}=:\rho^{\text{(eq)}}_{\zeta}(\varphi,\Pi,\tau,P)\ ,
\end{align}
with
\begin{align}
Z(b_\tau)&=\int_\Upsilon\mathcal D[\varphi,\Pi,\tau,P]\exp{\Bigl(-\braket{b_\tau,\tilde{\mathcal J}_H(\varphi,\Pi,\tau,P)}\Bigr)}\nonumber\\
&=\int_\Upsilon\mathcal D[\varphi,\Pi,\tau,P]\exp{\Bigl(-\braket{\zeta,\mathcal J_\tau(\varphi,\Pi,\tau,P)}\Bigr)}=Z(\zeta)\;,
\end{align}
where $\zeta$ is now restricted to the subset  $\Omega_\tau \subset \mathfrak{g}_\tau=\textsf{diff}(\Sigma_\tau)$ such that the above integral \emph{converge}. Due to the equivariance of the momentum map $\mathcal J_\tau$ under the action \eqref{stdiffsdiff}, the Gibbs statistical state \eqref{spatialdiffeqstate} is then an equilibrium state w.r.t. the one-parameter family of spatial diffeomorphisms generated by the vector field \eqref{genvecfields} associated to $\zeta\in\textsf{diff}(\Sigma_\tau)$. 

\begin{remark}
Note that, since in this reduced setting the embedding dependence of the vector field $\xi_{\mathcal{X}}(\tau(\vec{x}))$ given in \eqref{genvecfields} is reabsorbed into the projected constraint densities (cfr. Eq. \eqref{spatialmmap}), the state in \eqref{spatialdiffeqstate} can be derived directly via the prescriptions of Lie Group Thermodynamics with a $\textsf{diff}(\Sigma_{\tau})$-valued Lagrange multiplier entering the maximum entropy principle and playing the role of local temperature.
\end{remark}

\bibliographystyle{JHEP}


\end{document}